\documentclass[11pt,letterpaper,onecolumn,oneside,notitlepage]{article}

\usepackage{setspace}
\usepackage[letterpaper, margin=1in]{geometry}

\usepackage{amsmath}
\usepackage{amsthm}
\usepackage{amssymb}
\usepackage{amsfonts}
\usepackage{url}
\usepackage{xcolor}
\usepackage{graphicx}
\usepackage{subcaption}

\usepackage{hyperref}
\hypersetup{
    colorlinks=true,
    linkcolor=red,
    citecolor=blue,
    filecolor=magenta,      
    urlcolor=blue,
}

\usepackage[sort&compress,square,comma,authoryear]{natbib}

\newtheorem{theorem}{Theorem}
\newtheorem{lemma}[theorem]{Lemma}

\newtheorem{proposition}[theorem]{Proposition}
\newtheorem{definition}[theorem]{Definition}

\newtheorem{example}[theorem]{Example}
\newtheorem{counterexample}[theorem]{Counterexample}
\newtheorem{remark}[theorem]{Remark}

\newcommand{\INf}{\mathrm{In}_{f}}
\newcommand{\INnf}{\mathrm{In}_{-f}}
\newcommand{\rk}{\mathrm{ord}_{f}}
\newcommand{\exit}{\mathrm{Exit}}

\newcommand{\IsEmpty}{\mathrm{NonEmpty}}
\newcommand{\Neg}{\mathrm{Neg}}
\newcommand{\Reduce}{\mathrm{Reduce}}

\newcommand{\escape}{\mathrm{Escape}}
\newcommand{\rem}{\mathrm{rem}}
\newcommand{\deff}{\,\overset{\scriptscriptstyle\mathrm{def}}{=}\,}
\newcommand{\bbr}{\mathbb{R}}

\newcommand{\T}{\mathbf{T}}
\newcommand{\F}{\mathbf{F}}

\newcommand{\LZZ}{\textsf{\bfseries LZZ}}
\newcommand{\ExitSet}{\textsf{\bfseries ESE}}

\begin{document}

\title{Characterizing Positively Invariant Sets:\\
\Large{Inductive and Topological Methods}
}

\author{Khalil Ghorbal \\ \href{mailto:khalil.ghorbal@inria.fr}{khalil.ghorbal@inria.fr}
   \and Andrew Sogokon \\ \href{mailto:a.sogokon@soton.ac.uk}{a.sogokon@soton.ac.uk} }

\maketitle

\begin{abstract}
We present two characterizations of positive invariance of sets under the flow of systems of ordinary differential equations. %
The first characterization uses \emph{inward sets} which intuitively collect those points from which the flow evolves within the set for a short period of time,  whereas the second characterization uses the  notion of \emph{exit sets}, which intuitively collect those points from which the flow immediately leaves the set. 
Our proofs emphasize the use of the \emph{real induction} principle as a generic and unifying proof technique that captures the essence of the formal reasoning justifying our results and provides cleaner alternative proofs of known results.
The two characterizations presented in this article, while essentially equivalent, lead to two rather different decision procedures (termed respectively \LZZ{}~and~\ExitSet{}) for checking whether a given semi-algebraic set is positively invariant under the flow of a system of polynomial ordinary differential equations. 
The procedure \LZZ{} improves upon the original work by Liu, Zhan and Zhao~\citep{DBLP:conf/emsoft/LiuZZ11}. 
The procedure \ExitSet{}, introduced in this article, works by splitting the problem, in a principled way, into simpler sub-problems that are easier to check, and is shown to exhibit substantially better performance compared to \LZZ{} on problems featuring semi-algebraic sets described by formulas with non-trivial Boolean structure. 
\end{abstract}

\section{Introduction}

Positive invariance is an important concept in the theory of dynamical systems and one which also has practical applications in areas of computer science, such as formal verification, as well as in control theory. 
Informally, a set is \emph{positively invariant} if it is preserved under the evolution of the system according to the dynamics as time advances. 
A considerable amount of literature is dedicated to this subject~\citep{DBLP:journals/automatica/Blanchini99}, and great progress has been made in understanding positively invariant sets in continuous dynamical systems. 

In computer science, the notion of an \emph{inductive invariant} is analogous to that of a positively invariant set. 
It has relatively recently become the focus of considerable research interest, especially in the area of so-called \emph{hybrid systems}, which studies systems that combine discrete and continuous dynamics. Significant progress has been made over the past decade in the methods for algorithmically checking inductive invariants of ODEs (i.e. deciding whether a given set is positively invariant); these methods provide powerful tools for reasoning about the temporal behaviour of ODEs without the need to explicitly solve them. For example, one may use an inductive invariant to \emph{prove} that a system cannot evolve from a given set of initial conditions into a state which is deemed undesirable or unsafe (e.g. if the ODEs describe the motion of physical objects, one may wish to know that there can be no collisions between these objects in the future).

\paragraph{Contributions} This article presents a self-contained development of two characterizations of positively invariant sets of continuous systems (in Theorem~\ref{thm:lzz} and Theorem~\ref{thm:char}). 
In the case of semi-algebraic sets and polynomial ODEs, the two characterizations lend themselves to two alternative decision procedures for checking set positive invariance, both of which are described in detail. 

Section~\ref{sec:realinduction} is entirely devoted to the first characterization (Theorem~\ref{thm:lzz}), which relies on the concept of \emph{inward sets}~\cite[see][Def. 9.4]{DBLP:books/sp/17/ZWZ2017} and is very closely related to a theorem~\cite[Thm. 9.1]{DBLP:books/sp/17/ZWZ2017} which originally appeared in~\citep{DBLP:conf/emsoft/LiuZZ11}; we show how \emph{real induction}, via equivalent, yet subtly different formulations, can be used to give clean proofs of this known result and of Theorem~\ref{thm:lzz}. 
The section then describes a robust implementation of the associated decision procedure (\LZZ, after Liu, Zhan and Zhao), 
along with our improvements to the original method.

Section~\ref{sec:topology} presents the second characterization (Theorem~\ref{thm:char}), which is based on Conley's notion of \emph{exit sets}~\citep{conley}. 
We give a direct proof for this new result while formally establishing the relationship between exit sets and inward sets. 
Section~\ref{sec:complexity} presents a new algorithm (\ExitSet, which stands for \emph{Exit Set Emptiness}) that can more efficiently decide positive invariance of semi-algebraic sets described by formulas with non-trivial Boolean structure. The procedure works by splitting the problem 
into simpler sub-problems that are easier to check, reminiscent of divide-and-conquer algorithms. 

Our implementations of the two decision procedures $\LZZ{}$ and $\ExitSet{}$ are empirically evaluated on a number of positive invariance checking problems where semi-algebraic sets are described by non-atomic formulas in Section~\ref{sec:ex}, with \ExitSet{} exhibiting substantially better performance.
\section{Preliminaries}
A system of autonomous ordinary differential equations (ODEs) has the form:
\begin{align*}
x_1' &= f_1(x_1,\dots,x_n)\,,\\ 
\vdots & \\
x_n' &= f_n(x_1,\dots,x_n)\,,
\end{align*}
where $x_i'$ stands for the time derivative $\frac{d x_i}{dt}$ and 
$f = (f_1,f_2,\dots,f_n)$ is a vector-valued continuous function (which defines a \emph{vector field} on $\mathbb{R}^n$); we will write such a system more concisely as $x'=f(x)$.
We will denote by $\varphi(\cdot, x)$ the solution to the initial value problem $x'=f(x)$, with initial value $x \in \mathbb{R}^n$. 
We will only consider systems in which solutions to initial value problems always exist (at least locally) and are unique (e.g. local Lipschitz continuity of $f$ is sufficient to guarantee this property). 
When we quantify solutions over time $t$, we only consider $t$ in the \emph{maximal interval of existence} $I_{x}$, which in our case exists for any $x$ and contains $0$. 
In order to simplify our presentation, we will quantify over ``all forward time'' by writing  $\forall~t\geq0$ with the understanding that $\varphi(\cdot,x)$ may only be defined for $t \in I_{x}$. We refer to the mapping $\varphi$ as the (local) \emph{flow} of the vector field $f$.
\begin{definition}[Positively invariant set]
\label{def:posinv}
Given a system of ODEs $x'=f(x)$, a set $S\subseteq \mathbb{R}^n$ is \emph{positively invariant} if and only if no solution starting inside $S$ can leave $S$ in the future, i.e. just when the following holds:
\[
\forall~x\in S.~\forall~t\geq 0.~\varphi(t,x)\in S\,.
\]
\end{definition}

One analogously arrives at a definition of \emph{negatively invariant sets} in which no solution starting inside the set $S$ is permitted to be outside the set in the \emph{past}. Basic results in the theory of dynamical systems establish that a set $S$ is positively invariant precisely when its complement $S^c$ is negatively invariant \cite[Thm. 1.4]{BhatiaSzego1970}, and that the closure of a positively invariant set is also positively invariant~\cite[Prop. 1.4.5]{AlongiNelson2007}; this property also holds for the set's interior~\cite[Thm. 1.7]{BhatiaSzego1970}.

\begin{remark}
Some authors~\citep{BlanchiniMiani2008} favour a definition of positively invariant sets in which the solutions $\varphi(t,x)$ are explicitly required to exist for all time $t\geq 0$, by imposing a global Lipschitz continuity requirement on the vector field $f$, whereas others~\citep{redheffer1972theorems} simply require that solutions emanating from the set $S$ remain inside $S$ for as long as they exist in the future (Definition~\ref{def:posinv} is stated in this spirit).
\end{remark}

The first necessary and sufficient condition (i.e. characterization) for positive invariance of \textbf{closed} sets in systems of ODEs with unique solutions (but without requiring knowledge of the solutions $\varphi$) was given by~\cite{Nagumo1942},\footnote{Nagumo's result was in fact a little more general in that it did \emph{not} require unique solutions and focused on so-called \emph{weak positive invariance}, which is identical to positive invariance when solutions are unique.} and was later independently found by numerous other mathematicians (the interested reader is invited to consult~\citep{ DBLP:journals/automatica/Blanchini99}, \cite[Ch. 4, \S 4.2]{BlanchiniMiani2008}, and~\cite[Ch. III, \S 10, XV, XVI]{Walter1998} for more details about Nagumo's theorem and its multiple rediscoveries). Informally, Nagumo's theorem states that a \emph{closed set} $S$ is positively invariant if and only if at each point $x$ on the \emph{boundary} of $S$ the vector $f(x)$ points into the interior of the set or is tangent to it. The theorem may be easily applied in cases where the set $S$ is a \emph{sub-level set} of a differentiable real-valued function $g$, i.e. a set defined as $\{ x \in \bbr^n \mid g(x) \leq 0 \}$, provided that the gradient vector $\nabla g(x)$ is non-vanishing (i.e. non-zero) whenever $g(x)=0$ (intuitively this ensures that the boundary of $S$ is smooth): in this special case Nagumo's theorem says that $S$ is positively invariant if and only if $g'(x)\leq 0$ for all $x$ such that $g(x)=0$, where $g'$ denotes the (first) Lie derivative of $g$ with respect to the vector field $f$, which is defined by:~\footnote{The Lie derivative of $g$ is sometimes also denoted by $L_f(g)$ instead of $g'$.} 
\[ g' \deff \nabla g \cdot f = \sum_{i=1}^n \frac{\partial g}{\partial x_n}f_i~\,.
\]

\begin{remark}
Applying Nagumo's theorem in practice becomes problematic when the boundary of $S$ is not smooth, e.g. when the set \mbox{$\{ x \in \mathbb{R}^n \mid g(x)=0 \}$} contains \emph{singularities} (points $x$ where the gradient vanishes, i.e. $\nabla g(x)=0$); these issues have been explored by~\cite{DBLP:conf/fsttcs/TalyT09}.
In order to apply the theorem more generally to sets that are intersections of sub-level sets, i.e. \mbox{$\{ x \in \mathbb{R}^n \mid g_i(x) \leq 0,~i=1,\dots,k \}$}, one likewise needs to be very careful. The concept of \emph{practical sets} was introduced specifically to deal with these issues~\citep[see][Ch. 4, Def. 4.9]{BlanchiniMiani2008}. 
\end{remark}

In the following sections we will be concerned with characterizations of positive invariance that are of a very different nature to that of Nagumo's result (which provides a characterization only for closed sets obtained using the tools of real analysis and is without effective computational means of applying it). 
As we shall see, these alternative characterizations can be effectively applied using tools from commutative algebra and real algebraic geometry. 
\section{Characterizing Positive Invariance Through Inward Sets}
\label{sec:realinduction}
Let us consider the following construction of the so-called \emph{inward set} for a given set $S\subseteq \bbr^n$ 
and a system of ODEs ${x}' = f({x})$ for which a unique (local) solution to the initial value problem exists for any $x \in \bbr^n$, following~\cite[Def. 9.4]{DBLP:books/sp/17/ZWZ2017}:
\[\INf(S) \deff \{ x \in \mathbb{R}^n \mid \exists~\varepsilon >0.~\forall~t\in (0,\varepsilon).~\varphi(t,x) \in S \}\,. \]
When time/flow is reversed, one can likewise construct the \emph{inverse inward set}:
\[\INnf(S) = \{ x \in \mathbb{R}^n \mid \exists~\varepsilon >0.~\forall~t\in (0,\varepsilon).~\varphi(-t,x) \in S \}\,. \]
It is useful to intuitively think of these as sets of states from which the system will evolve inside $S$ for some non-trivial time interval ``immediately in the future'' and, respectively, has evolved inside $S$ for some non-trivial time interval ``immediately in the past''. 

We observe that, although according to the statement of $\INf(S)$, $x$ can be any point in $\bbr^n$, it is in fact restricted to the closure of $S$.  
\begin{lemma}
\label{lem:infSinClosure}
For $S \subseteq \bbr^n$, the set $\INf(S)$ is a subset of the closure of $S$. 
\end{lemma}
\begin{proof}
We show by contradiction that $\INf(S) \cap (S^c)^\circ$ is empty. 
Let $x$ be an element of the intersection.  
Since $x \in (S^c)^\circ$, there exists an open neighbourhood $U \subset (S^c)^\circ$ of $x$ and $\rho > 0$ such that $\varphi(t,x) \in U \subset S^c$ for all $t \in (0,\rho)$ (by continuity of $\varphi(\cdot,x)$). 
But if $x \in \INf(S)$, then there exists $\varepsilon > 0$ such that $\varphi(t,x) \in S$ for all $t \in (0,\varepsilon)$. 
So $\varphi(t,x)$ is in both $S^c$ and $S$ for all $t \in (0, \min\{\rho,\varepsilon\})$, a contradiction. 
\end{proof}

\begin{remark}
\label{rm:infinclusions}
Notice that the \emph{interior} of $S$ is always contained inside $\INf(S)$ (by definition) and the inclusion $S \subseteq \INf(S)$ therefore holds trivially whenever $S$ is an open set (for any $f$).  %
A quick glance at the definitions
\[\INf(S^c) = \{ x \in \mathbb{R}^n \mid \exists~\varepsilon >0.~\forall~t\in (0,\varepsilon).~{\varphi}(t,x) \not\in S \}\,, \]
\[\INf(S)^c = \{ x \in \mathbb{R}^n \mid \forall~\varepsilon >0.~\exists~t\in (0,\varepsilon).~{\varphi}(t,x) \not\in S \}\,, \]
reveals the following inclusion $\INf(S^c) \subseteq \INf(S)^c$ for any set $S$. %
Whenever $S$ is a closed set ($S^c$ is open), we therefore have that $S^c \subseteq \INf(S^c) \subseteq \INf(S)^c$ %
or, if we prefer, $\INf(S) \subseteq S$. %
\end{remark}

These constructions can be used to state the following characterization of positively invariant sets.

\begin{theorem}
A set $S\subseteq \mathbb{R}^n$ is positively invariant under the flow of the system ${x}' = f({x})$ if and only if $S \subseteq \INf(S)$ and $S^c \subseteq \INnf(S^c)$.
\label{thm:lzz}
\label{thm:charbyinduction}
\end{theorem}

Theorem~\ref{thm:lzz} can be understood and proved using \emph{induction over the non-negative real numbers}.
Though there are many different variations of induction over the reals~\cite[e.g. see][]{clark2019}, this method of proof appears to be far less well known than standard mathematical induction over the natural numbers. 
We state below a version of real induction that is well suited to directly prove the theorem. 

\begin{lemma}[Real induction]
A predicate $P(t)$ holds for all $t\geq 0$ if and only if: 
\begin{enumerate}
\item $P(0)$\,,
\item[\emph{2.}]   $\forall~t \geq 0. ~\Big(\lnot P(t) \to \big( \exists~\varepsilon > 0.~\forall~T\in (t - \varepsilon, t).~\lnot P(T) \big) \Big)$\,,
\item[\emph{3.}]   $\forall~t \geq 0.~ \Big( P(t) \to \big( \exists~\varepsilon > 0.~\forall~T\in (t, t + \varepsilon).~P(T) \big) \Big)$\,.
\end{enumerate}
\label{thm:realind}
\end{lemma}

\begin{proof}
Necessity is obvious. Sufficiency is easy to show by considering (for contradiction) that there exists $t \geq 0$ such that $\lnot P(t)$ and defining the time $t_*= \inf \{t \geq 0 \mid \lnot P(t)\}$ (which exists as the set is assumed to be non-empty, is bounded from below, and the reals are complete). %
By 1. and 3. we have that $t_* \neq 0$, so $t_*$ must be positive, but in this case $P(t)$ holds for all $t\in [0,t_*)$ (by definition). If $P(t_*)$ holds, then $t_*$ cannot be an infimum (by 3.), and if $\lnot P(t_*)$ then (by 2.) we have that $\lnot P(t)$ holds for all $t \in (t_* - \varepsilon, t_*)$ for some $\varepsilon>0$; a contradiction.
\end{proof}

Using the above real induction principle, the proof of Theorem~\ref{thm:lzz} is immediate if one takes ``$\varphi(t,x) \in S$'' to be the predicate $P(t)$ in Lemma~\ref{thm:realind}. 
We remark that (unlike Nagumo's theorem), Theorem~\ref{thm:lzz} makes no assumptions about the set $S$ being closed, or open. As such, Theorem~\ref{thm:charbyinduction} is very general and applies to all sets and systems of ODEs with locally unique solutions.

Theorem~\ref{thm:lzz} is closely related to~\cite[Thm. 9.1]{DBLP:books/sp/17/ZWZ2017} 
where the authors require $S^c \subseteq \INnf(S)^c$ instead of $S^c \subseteq \INnf(S^c)$.\,\footnote{The set inclusions required in~\cite[Thm. 9.1]{DBLP:books/sp/17/ZWZ2017} can be alternatively phrased as \mbox{$\INnf(S)\subseteq S \subseteq \INf(S)$}.}
Despite the fact that, in general, $\INnf(S)^c \neq \INnf(S^c)$ (cf. counterexample~\ref{unioncounterexample}), the conditions in Theorem~\ref{thm:lzz} and~\cite[Thm. 9.1]{DBLP:books/sp/17/ZWZ2017} are in fact equivalent. We show this equivalence by appealing again to real induction: we first state a slightly different real induction principle that is more suited to prove~\cite[Thm. 9.1]{DBLP:books/sp/17/ZWZ2017}, providing thereby a new simpler proof for this known result, and then show that both principles are in fact equivalent.~\footnote{The idea of using real induction to prove~\cite[Thm. 9.1]{DBLP:books/sp/17/ZWZ2017} was first suggested by Paul B. Jackson and Kousha Etessami (School of Informatics, University of Edinburgh) in private communication with the second author.} 

\begin{lemma}[Real induction (Jackson)]
A predicate $P(t)$ holds for all $t\geq 0$ if and only if: 
\begin{enumerate}
\item $P(0)$,
\item[\emph{2'.}]   $\forall~t > 0.~ \Big( \big( \exists~\varepsilon \in (0,t].~  \forall~T \in (t - \varepsilon, t).~ P(T) \big) \to P(t) \Big)$,
\item[\emph{3.}]   $\forall~t \geq 0.~ \Big( P(t) \to  \big( \exists~\varepsilon > 0.~\forall~T\in (t, t + \varepsilon).~P(T) \big) \Big)$.
\end{enumerate}
Notice that condition \emph{2'.} could be equivalently replaced by its contrapositive form 
\begin{itemize}
    \item[\emph{2''.}] $\forall~t > 0.~ \Big(\lnot P(t) \to  \big( \forall~\varepsilon \in (0,t].~\exists~T\in (t - \varepsilon, t).~\lnot P(T) \big) \Big)$\,.
\end{itemize}
\label{thm:realindLZZ}
\end{lemma}
\begin{proof}
Necessity is obvious. Sufficiency is easy to show by considering (for contradiction) the time $t_*= \inf \{t \geq 0 \mid \lnot P(t)\}$. By 1. and 3. we have that $t_* \neq 0$, so $t_*$ must be positive, but in this case $P(t)$ holds for all $t\in [0,t_*)$ (by definition) and by 2'. we have that $P(t_*)$ holds; a contradiction.
\end{proof}

\begin{remark}
For completeness, we include below a statement of Hathaway's \emph{continuity induction}~\citep{hathaway2011}, which is very similar to the notion of real induction in~\citep{clark2019}. 
A predicate $P(t)$ holds for all $t \in [0, T]$, where $T >0$, if and only if: 
\begin{enumerate}
\item $P(0)$ holds\,,
\item 
$\forall~\tau \in (0,T].~\Big( \big( \forall~\tau'\in[0,\tau).~P(\tau')\big) \to P(\tau) \Big)$\,,
\item $\forall~\tau \in [0,T).~\Big(\big(\forall~\tau' \in [0,\tau].~P(\tau')\big)  \to \big(\exists~\epsilon>0.~\forall~\tau'' \in (\tau,\tau+\epsilon).~P(\tau'')\big)\Big)$\,.
\end{enumerate}
\label{thm:contind}

The proof is essentially identical to that of Lemma~\ref{thm:realindLZZ}.

\end{remark}

The proof of~\cite[Thm. 9.1]{DBLP:books/sp/17/ZWZ2017} now becomes immediate using real induction as it is stated in Lemma~\ref{thm:realindLZZ}. 
The following lemma establishes an equivalence between the two formulations of real induction.

\begin{lemma} \label{lem:eqLemmas} 
Let $P(t)$ denote a predicate defined for all $t \geq 0$. If
\begin{enumerate}
\item[\emph{1.}] $P(0)$\,, and 
\item[\emph{3.}]   $\forall~t \geq 0.~ \Big( P(t) \to \big( \exists~\varepsilon > 0.~\forall~T\in (t, t + \varepsilon).~P(T) \big) \Big)$\, hold ,
\end{enumerate}
then 
\begin{enumerate}
\item[\emph{2.}]   $\forall~t \geq 0.~ \Big( \lnot P(t) \to  \big( \exists~\varepsilon > 0.~\forall~T\in (t - \varepsilon, t).~\lnot P(T) \big) \Big)$\,,
\end{enumerate}
if and only if 
\begin{enumerate}
\item[\emph{2''.}]   $\forall~t > 0.~ \Big(\lnot P(t) \to  \big( \forall~\varepsilon \in (0,t].~\exists~T\in (t - \varepsilon, t).~\lnot P(T) \big) \Big)$\,.
\end{enumerate}
\end{lemma}
\begin{proof}
The implication from 2. to 2''. is obvious (in this sense, one may consider Lemma~\ref{thm:realind} \emph{weaker} than Lemma~\ref{thm:realindLZZ}). To prove the converse, suppose (for contradiction) that 2''. and $\lnot 2.$ both hold. 
More explicitly: 
\[
\lnot 2.~\equiv~\exists~t \geq 0.~ \Big(\lnot P(t) \land  \big( \forall~\varepsilon >0.~\exists~T\in (t - \varepsilon, t).~P(T) \big) \Big)\,.
\]

Let $\tau>0$ be the point at which $\lnot P(\tau)$ holds in $\lnot 2.$ ($\tau$ cannot be $0$ by $1.$) Then for all $\varepsilon_0>0$, there exists some $T_0\in(\tau-\varepsilon_0, \tau)$ such that $P(T_0)$ holds. Consider the interval $I_0 = [T_0,\tau]$. At $T_0$, since $P(T_0)$ holds, we have (by 3.) that $P(t)$ holds for all $t$ in the interval $I_1 = [T_0,T_0+\varepsilon_1)$ for some $\varepsilon_1>0$. If $T_0+\varepsilon_1 \geq \tau$, we obtain a contradiction (because $\lnot P(\tau)$ is assumed to hold); otherwise we have $I_1=[T_0,T_0+\varepsilon_1) \subset [T_0, \tau]$. If at the endpoint of $I_1$ we have that $\lnot P(T_0+\varepsilon_1)$ holds, we obtain a contradiction (by 2''.), and if $P(T_0+\varepsilon_1)$ holds we have (by 3.) that for some $\varepsilon_2>0$, $P(t)$ holds for all $t \in I_2 =  [T_0, T_0+\varepsilon_1+\varepsilon_2)$. Repeating the argument, we obtain a sequence $I_k$ of intervals of strictly increasing length where $P(t)$ holds. The right endpoints of the intervals in this sequence cannot converge within $(T_0,\tau]$ because this would yield a contradiction (by 2''.). 
The right endpoints thus go beyond $\tau$, which again yields a contradiction.
\end{proof}

The equivalence stated in Lemma~\ref{lem:eqLemmas} is somewhat abstract and it may not be immediately clear how this equivalence is relevant with regard to inward sets. To make this more apparent, we prove below a lemma which can be used to establish the equivalence between Theorem~\ref{thm:lzz} and~\cite[Thm. 9.1]{DBLP:books/sp/17/ZWZ2017} without appealing to real induction directly in the proof, although following a similar line of argument as that employed in the proof of Lemma~\ref{lem:eqLemmas}. 
\begin{lemma}
\label{lem:subtlelink}
Let $S\subseteq \bbr^n$. If $S \subseteq \INf(S)$ then $\INnf(S) = \INnf(S^c)^c$. 
\end{lemma}
\begin{proof}
The inclusion $\INnf(S) \subseteq \INnf(S^c)^c$ holds in general by definition as already stated. 
Let $x \in \INnf(S^c)^c$ and let $\varepsilon_0 > 0$. Then, by definition, there exists $t_0 \in (0,\varepsilon_0)$, such that $x_0 :=  {\varphi}(-t_0,x) \not\in S^c$, or equivalently $x_0 \in S$. 
Since $S \subseteq \INf(S)$, $x_0 \in \INf(S)$ and there exists $\gamma_0 > 0$, such that for all $s_0 \in (0,\gamma_0)$, $\varphi(s_0,x_0) \in S$. 
If $-t_0+\gamma_0 < 0$, then the same arguments  with $\varepsilon_1 := t_0 - \gamma_0$ lead to the existence of $t_1,\gamma_1 > 0$ such that for all $s_1 \in (0,\gamma_1)$, $\varphi(s_1,x_1) \in S$ where $x_1 := {\varphi}(-t_1,x)$. 
We can thus construct a (strictly) increasing sequence $-t_0+\gamma_0, -t_1+\gamma_1, \dotsc$. 
Two cases may occur: 

(i) If the sequence crosses zero after finitely many steps, that is there exists $n\geq 0$ such that $-t_n+\gamma_n \geq 0$, then this means that for all $t \in (0,t_n)$, $\varphi(-t,x) \in S$ thereby proving that $x \in \INnf(S)$ since $-t_n < -t < 0 \leq -t_n +\gamma_n$. 

(ii) If the sequence is upper bounded by $0$, then it has a limit $-t_l+\gamma_l \leq 0$. 
The case $-t_l+\gamma_l < 0$ is impossible since we can perform one more step to get $-t_l+\gamma_l < -t_{l'}+\gamma_{l'} \leq 0$. 
Thus $-t_l+\gamma_l = 0$ and one gets $\varphi(-t,x) \in S$ for all $t\in (0,t_l)$ leading, as in case (i), to $x \in \INnf(S)$. 
\end{proof}

\begin{remark}
The statement of the characterization in Theorem~\ref{thm:lzz} enjoys some rather nice properties when compared to that of~\cite[Thm. 9.1]{DBLP:books/sp/17/ZWZ2017}. It is in particular symmetric in the sense that the set inclusions in the theorem are preserved when one simultaneously replaces $S$ with its complement $S^c$ and $f$ with the reversed dynamics $-f$. 
This allows for instance to immediately prove the well-known result in dynamical systems which states that a set is positively invariant if and only if its complement is negatively invariant~\cite[Thm. 1.4]{BhatiaSzego1970}. 
One sees that by syntactically replacing $S$ with $S^c$ and $f$ with $-f$ in the conditions of Theorem~\ref{thm:lzz}, one obtains $S^c \subseteq \INnf(S^c)$ and $S\subseteq \INf(S)$, i.e. \emph{equivalent} conditions, using only the set-theoretic fact that $(S^c)^c = S$. 
On the other hand, applying the same transformation to the conditions $S\subseteq \INf(S)$ and $S^c \subseteq \INnf(S)^c$ required in~\cite[Thm. 9.1]{DBLP:books/sp/17/ZWZ2017}, one does \emph{not} immediately obtain the same conditions; instead, one obtains \mbox{$S^c\subseteq \INnf(S^c)$} and \mbox{$S \subseteq \INf(S^c)^c$}. In order to show that the original inclusions hold one needs to use the fact that $\INnf(S^c) \subseteq \INnf(S)^c$ for the first inclusion, and then use Lemma~\ref{lem:subtlelink} for the second inclusion, which is somewhat more involved than using Theorem~\ref{thm:lzz} to prove the same fact.
\end{remark}

The main practical difficulty in applying Theorem~\ref{thm:lzz} (or equivalently ~\cite[Thm. 9.1]{DBLP:books/sp/17/ZWZ2017}) lies in the fact that inward sets $\INf(S)$ and $\INnf(S^c)$ are defined in terms of solutions to a system of differential equations; the theorem says nothing about our ability to construct these sets or reason about their inclusion. 
The following section will elucidate how this problem is addressed using tools from algebraic geometry in the important case where the set $S$ is semi-algebraic and the right-hand side of the system $x'=f(x)$ is polynomial.

\subsection{A Decision Procedure for Checking Positively Invariant Sets}
\label{sec:computation}
In this section we describe a procedure for deciding whether a given set is positively invariant or not. For this we first require a few basic results. Let $g: \bbr^n \to \bbr$ denote a real-valued function. The zero-th Lie derivative of $g$ is $g$ itself, the first order Lie derivative $g' \deff \nabla g \cdot f$ corresponds to the total derivative of $t \mapsto g(\varphi(t,x))$ with respect to time $t$, and higher-order Lie derivatives are defined inductively, i.e. $g'' = (g')'$; the $k$-th order Lie derivative of $g$ will be denoted by $g^{(k)}$. 
We will require the fact that unique solutions to real analytic systems of ODEs are also real analytic~\cite[Thm 1.3]{Chicone}. 
Whenever $g$ is a real analytic function, its Taylor series expansion
\[
g(\varphi(t,x)) = g(x) + g'(x)t + g''(x)\frac{t^2}{2!} + \cdots
\]
converges in some time interval $(\epsilon_l, \epsilon_u)$, where $\epsilon_l<0<\epsilon_u$.
The set of states $\{ x \in \mathbb{R}^n \mid g(x)=0 \}$, simply denoted by $g=0$ in the sequel, remains invariant under the flow for some non-trivial forward time interval if and only if \emph{all} Lie derivatives $g^{(k)}$, $k \geq 1$, vanish whenever $g(x)=0$. 
\begin{remark}
\label{rmk:notation}
We will abuse notation slightly in this article by interchangeably using sets and formulas characterizing those sets. For example, we will use formulas in the arguments to $\INf$ and $\INnf$ (from Theorem~\ref{thm:charbyinduction}). However, when describing sets we will use set-theoretic symbols $\cup$ and $\cap$ for set union and intersection, respectively, and will let $S^c$ denote the complement of $S$; when we are working with formulas, we will instead employ the corresponding logical symbols $\lor$ and $\land$ for disjunction and conjunction, and $\lnot$ for negation. 
The set $\bbr^n$ (resp. $\emptyset$) will be syntactically represented by the symbol $\T$ (resp. $\F$).
\end{remark}
We thus have the inward set of $g=0$ given by 
\begin{align*}
\quad \INf(g = 0)  & \quad = \quad  g=0 \cap g'=0 \cap g''=0 \cap  g'''=0 \cap \cdots  \quad\,,
\end{align*}
which is characterized by the following infinite ``formula''~\footnote{Technically, a formula can only be finite, hence the quotes for such hypothetical objects.} %
\begin{align*}
 \text{``} \quad \INf(g = 0)  & \quad \equiv \quad  g=0 \land g'=0 \land g''=0 \land  g'''=0 \land \cdots  \quad \text{''}\,.
\end{align*}

For sets of states satisfying inequalities $\{ x \in \mathbb{R}^n \mid g(x)<0 \}$, which we also concisely denote by the formula $g<0$, the situation is similar with the following infinite construction: 
\begin{align*}
 \text{``} \quad \INf(g < 0)  & \quad \equiv \quad  g<0  \\
 \quad & \quad \lor~(g=0 \land g'<0) \\
 \quad & \quad \lor~(g=0 \land g'=0 \land g''<0) \\
 \quad & \quad \lor~(g=0 \land g'=0 \land g''=0 \land  g'''<0) \\
 \quad & \quad  \ \vdots \\
 \quad & \quad \quad \quad \quad \text{''}\,.
\end{align*}
Intuitively, the first non-zero Lie derivative of $g$ needs to be negative at a point $x$ satisfying $g(x)=0$ in order for the flow ${\varphi}(t,x)$ to enter the set $g<0$ from that point and remain within this set throughout some time interval $(0,\epsilon)$, for some positive $\epsilon$.
\begin{remark}
One may draw physical analogies here, e.g. to the motion of a vehicle: if the velocity is $0$, then it is the sign of the acceleration term that determines whether the vehicle will move forward in the next time instant; if both the velocity and the acceleration are $0$, it is the sign of the derivative of the acceleration (i.e. the sign of the jerk term), and so forth. 
\end{remark}

The decision procedure developed by~\cite{DBLP:conf/emsoft/LiuZZ11} rests on the fact that for a polynomial function $p$ and a polynomial system of ODEs $x'=f(x)$, the formulas characterizing 
\mbox{$\INf(p=0)$} and \mbox{$\INf(p<0)$} are indeed \emph{finite}. 
To see why this is true, note that whenever $p$ and $f_1,f_2,\dots,f_n$ that make up $f$ are polynomials, all the formal Lie derivatives $p', p'', \cdots$ are also guaranteed to be polynomials. 
Let us now recall the \emph{ascending chain property} of ideals in the polynomial ring $\mathbb{R}[x_1,\dots,x_n]$ -- a consequence of Hilbert's basis theorem and the fact that the ring $\mathbb{R}$ is Noetherian~\cite[Ch. 2, Thm. 7]{Cox2015}.
\begin{lemma}
\label{lemma:acc}
Let $p\in \mathbb{R}[x_1,\dots,x_n]$, then the ascending chain of ideals
\[
\langle p \rangle \subseteq \langle p, p' \rangle \subseteq \langle p, p', p'' \rangle \subseteq \cdots
\]
is finite, i.e. there exists a $k \in \mathbb{N}$ such that $\langle p, p', \dots, p^{(k)} \rangle = \langle p, p', \dots, p^{(K)} \rangle$ for all $K \geq k$. 
\end{lemma}
For a given $p$, we denote the smallest $k$ in the above lemma by $\rk(p)$ and say that it defines the \emph{order} of $p$ with respect to the system of polynomial ODEs $x'=f(x)$. 
In practice, we can always compute $\rk(p)$ by simply computing successive formal Lie derivatives of $p$ and successively checking whether
\[
p^{(k+1)} \in \langle  p, p', p'', \dots, p^{(k)} \rangle
\]
holds for $k=1,2,3,\dots$, until the membership check succeeds, which would imply that the ideal chain has stabilized (the fact that this process terminates is guaranteed by Lemma~\ref{lemma:acc}).\footnote{Using terminology from differential algebra~\citep{Ritt1950} one may say that the ideal \mbox{$\langle p, p', \dots, p^{(\rk(p))}\rangle$} defines a \emph{differential ideal}.}
The ideal membership check can be easily performed by reducing the polynomial $p^{(k+1)}$ by the Gr\"obner basis of $\{p, p', \dots, p^{(k)} \}$ for each successive $k$ and checking whether the remainder is $0$.
An upper bound on the length of the ascending chain of ideals generated by successive Lie derivatives of $p$ was obtained in~\cite[Thm. 4]{NovikovYakovenko}; this bound is doubly-exponential in the number of variables, however, in practice one typically observes the ideals stabilizing after only a few iterations.

As a direct consequence of Lemma~\ref{lemma:acc}, 
whenever 
\(
p,p',\dots,p^{(\rk(p))}
\) 
are all simultaneously $0$, all higher derivatives must also evaluate to $0$. 
More formally: 
\begin{align*}
p=0 \land p'=0 \land p''=0 \land  \dots \land p^{(\rk(p))}=0\ \to \  \forall~K>\rk(p).~p^{(K)}=0 \, .
\end{align*}
Using this fact one can construct perfectly legitimate formulas that provide a finite characterization of $\INf(p=0)$ and $\INf(p<0)$, given as follows:

\begin{align*}
 \quad \INf(p = 0)  & \quad \equiv \quad p=0 \land p'=0 \land p''=0 \land \dots \land p^{(\rk(p))}=0\,, \\
 \quad \INf(p < 0)  & \quad \equiv \quad  p<0  \\
 \quad & \quad \lor~(p=0 \land p'<0) \\
 \quad & \quad \lor~(p=0 \land p'=0 \land p''<0) \\
 \quad & \quad  \ \vdots \\
 \quad & \quad \lor~(p=0 \land p'=0 \land p''=0 \land \dots \land p^{(\rk(p))}<0)\,.
\end{align*} 
Notice that in the construction of $\INf(p<0)$ the saturation of the chain of ideals guarantees that all further terms in the disjunction, i.e.  
\[
p=0 \land p'=0 \land p''=0 \land \dots \land p^{(\rk(p))}=0 \land \dots \land p^{(K)}<0
\] where $K>\rk(p)$, are \textsf{False} and therefore unnecessary.

\subsection{Improving the Construction of $\INf(p=0)$ and $\INf(p<0)$}
\label{subsec:lzzimprovements}
One may work na\"ively with ideals generated by the successive Lie derivatives $\langle p, p', p'', \dots, p^{(k)} \rangle$ and construct $\INf(p=0)$ and $\INf(p<0)$ using these derivatives directly (as above), following~\cite{DBLP:conf/emsoft/LiuZZ11}. However, this construction can be improved if one realizes that only the \emph{remainders} of the Lie derivatives are needed for this construction, as will be shown in the following lemma. The practical advantage afforded by doing this is the degree of the remainder polynomials, which is typically lower than the degree of the Lie derivatives themselves. 

\begin{lemma}
\label{lem:rem}
Given a polynomial $p$ and a system of polynomial ODEs $x'=f(x)$, let $\rem_0 = p$ and let $\rem_{i+1}$ be defined inductively as the \emph{remainder} obtained from polynomial reduction (i.e. multivariate polynomial division) of the Lie derivative $\rem_{i}'$ by the polynomials $\{\rem_0, \rem_1 \dots, \rem_{i} \}$. Then for all $i \geq 0$
\[\rem_{i} = p^{(i)} - \sum_{j=0}^{i-1} \alpha_{ij} p^{(j)} \]
where $\alpha_{i j}$ are polynomials. %
\label{lemma:remsigmap}
\end{lemma}
\begin{proof}
By induction. Base case: $\rem_0 = p = p^{(0)}$.
For an inductive hypothesis, assume that \( \rem_k = p^{(k)} - \sum_{j=0}^{k-1} \alpha_{kj} p^{(j)} \) holds for all \mbox{$k \leq n$}. Since $\rem_{n+1}$ is the remainder upon the reduction of $\rem_{n}'$ by $\{ \rem_0, \dots, \rem_{n} \}$, we have \mbox{$\rem_{n+1} = \rem_n' - \sum_{i=0}^{n} \beta_i \rem_i$}, where $\beta_0,\dots,\beta_{n}$ are polynomials. From our inductive hypothesis and by applying the product rule for differentiation we have
\begin{align}
\rem_n' = p^{(n+1)} - \left(\sum_{j=0}^{n-1} \alpha_{nj} p^{(j)}\right)'  &= p^{(n+1)} - \sum_{j=0}^{n} \gamma_j p^{(j)}\,,
\label{lemma:remprime}
\end{align}
where $\gamma_0, \dots,\gamma_{n}$ are polynomials, and 
\begin{align*}
\rem_{n+1} &= \rem_n' - \sum_{i=0}^{n} \beta_i \rem_i \qquad \text{[from the definition]\,}\\
 &= \left(p^{(n+1)} - \sum_{j=0}^{n} \gamma_j p^{(j)}\right) - \sum_{i=0}^{n} \beta_i \rem_i \qquad \text{[from (\ref{lemma:remprime})]}\,, \\
 &= \left(p^{(n+1)} - \sum_{j=0}^{n} \gamma_j p^{(j)}\right) - \sum_{i=0}^{n} \beta_i\left(p^{(i)} - \sum_{l=0}^{i-1} \alpha_{il} p^{(l)}\right) \quad \text{[by hypothesis]\,,}
\end{align*}
from which it is apparent that $\rem_{n+1}$ has the required form: 
\[
\rem_{n+1} = p^{(n+1)} - \sum_{j=0}^{n} \alpha_{n+1j} p^{(j)}\,.
\]
\end{proof}

\begin{lemma}
\label{lem:INfrem}
Let $\rem_i$ be defined as in Lemma~\ref{lemma:remsigmap}. Then the inward sets can be characterized as follows:
\begin{align*}
 \quad \INf(p = 0)  & \quad \equiv \quad (\rem_0=0 \land \rem_1=0 \land \rem_2=0 \land \dots \land \rem_{\rk(p)}=0)
\end{align*}
and
\begin{align*}
 \quad \INf(p < 0)  & \quad \equiv \quad  \rem_0<0  \\
 \quad & \quad \lor~(\rem_0=0 \land \rem_1<0) \\
 \quad & \quad \lor~(\rem_0=0 \land \rem_1=0 \land \rem_2<0) \\
 \quad & \quad  \ \vdots \\
 \quad & \quad \lor~(\rem_0=0 \land \rem_1=0 \land \rem_2=0 \land \dots \land \rem_{\rk(p)}<0)\,.
\end{align*}
\end{lemma}
\begin{proof}
For $\INf(p=0)$, we show by induction that
\[
\forall~n\geq 0.~\left(\bigcap_{i=0}^n \rem_i=0\right)  \ = \  \left(\bigcap_{i=0}^n p^{(i)}=0 \right)\,.
\]
Base case: $\rem_0 = p^{(0)}=p$ by definition. For the inductive hypothesis, let us assume that 
\[
\left(\bigcap_{i=0}^k \rem_i=0\right)  \ = \  \left(\bigcap_{i=0}^k p^{(i)}=0 \right)\,
\]
holds for some $k\geq 0$. Then from the hypothesis we have that
\[
\left(\bigcap_{i=0}^{k+1} \rem_i=0\right)  \ = \  \left(\bigcap_{i=0}^{k} p^{(i)}=0~\cap \rem_{k+1}=0\right)\,.
\]
By Lemma~\ref{lemma:remsigmap} we have 
\(
\rem_{k+1} = p^{(k+1)} - \sum_{j=0}^{k} \alpha_{k+1j} p^{(j)}
\) and hence
\begin{align*}
\left(\bigcap_{i=0}^{k+1} \rem_i=0\right)  \ &= \  \left(\bigcap_{i=0}^{k} p^{(i)}=0~\cap p^{(k+1)} - \sum_{j=0}^{k} \alpha_{k+1j} p^{(j)}=0\right)\, \\
\ &= \ \left(\bigcap_{i=0}^{k} p^{(i)}=0~\cap p^{(k+1)}=0\right)\,.
\end{align*}
The proof for $\INf(p<0)$ follows a similar inductive argument.
\end{proof}

\begin{remark}
Using the remainders instead of the higher-order Lie derivatives of $p$ for constructing $\INf(p=0)$ and $\INf(p<0)$ is pragmatically often a good choice. 
For a concrete example, consider the Van der Pol oscillator whose dynamics is given by $x'=y$ and $y'=-x - y (x^2 - 1)$, and let $p=x^2+y^2-1$. The ascending chain of ideals \[
\langle \rem_0 \rangle \subseteq \langle \rem_0, \rem_1 \rangle \subseteq \langle \rem_0, \rem_1, \rem_2 \rangle \subseteq \cdots
\]
stabilizes at $\langle \rem_0, \dots, \rem_6 \rangle$, which is  $\langle x^2 + y^2 -1, 2 y^4, -8 x y^3, 24 y^2, -48 x y, 48 \rangle$. In contrast, if one uses the actual higher-order Lie derivatives, the $6$ generators of the ideal $\langle p, p', \dots, p^{(6)} \rangle$ are too large to all fit on this page, with $p^{(6)}$ having total degree $12$. 
We should however note that in certain cases it may be more expensive to compute the ideal $\langle \rem_0, \rem_1,\dotsc, \rem_{\rk(p)} \rangle$ than it is to compute $\langle p, p', \dotsc, p^{(\rk(p))} \rangle$ because the potential gain in lowering the total degree of the ideal generators can be outweighed by the computational overhead arising from the size of the coefficients of the intermediate polynomials. This is a well-known phenomenon when computing Gr\"obner bases~\cite[Ch. 2, p.116]{Cox2015}.  
\end{remark}

\subsection{Distributive Properties of $\INf$}
Viewed as a set operator, $\INf$ distributes over set intersections. For any sets $S_1, S_2 \subseteq \bbr^n$, one has: 
\[\INf(S_1 \cap S_2) = \INf(S_1) \cap \INf(S_2)\,. \]
The operator $\INf$ does not, however, distribute over set union; only the following set inclusion is guaranteed to hold in general:
\[
\INf(S_1 \cup S_2) \supseteq \INf(S_1) \cup \INf(S_2)\,.
\]

\begin{counterexample}\label{unioncounterexample}
To see why the converse inclusion does not hold, consider the simple $1$-dimensional system $x' = 1$ and the set    
\[
S = \left\{ x \in \bbr \mid x \leq 0 \lor \left(x > 0 \land \sin\left(x^{-1}\right)=0\right) \right\}\, .
\]
The point $0\in \bbr$ cannot be an element of $\INf(S)$ because $\varphi(t,0)=t$ and for any positive $\epsilon$ there exists a $t \in (0,\epsilon)$ such that $\sin\bigl(t^{-1}\bigr)\neq 0$ and therefore $\varphi(t,0) \not\in S$. 
In other words, $0$ belongs to $\INf(S)^c$. 
At the same time, $0$ cannot be in $\INf(S^c)$ either because the flow cannot move from the point at $x=0$ without crossing one root of $\sin\bigl(t^{-1}\bigr)=0$. 
Thus 
\[
\INf(S \cup S^c) = \INf(\mathbb{R}^n) = \bbr^n \neq \INf(S) \cup \INf(S^c) \, .
\]
The example also shows that in general $\INf(S^c)$ is not equal to $\INf(S)^c$ since $0\in \INf(S)^c$ while $0\not\in \INf(S^c)$.\,\footnote{Recall from Remark~\ref{rm:infinclusions} that $\INf(S^c) \subseteq \INf(S)^c$ holds for any set $S$. The above counterexample demonstrates that the converse inclusion does not hold generally.}

\end{counterexample}
For semi-analytic sets the $\INf$ operator \emph{does distribute} over set unions. In particular, for semi-algebraic sets (a special class of semi-analytic sets) given by
\[ 
S = \bigcup_{i=1}^{l}\left( \bigcap_{j=1}^{m_i}~p_{ij} < 0~\cap~\bigcap_{j=m_i+1}^{M_i}~p_{ij} = 0 \right)\,, 
\]
where $p_{ij}$ are polynomials, one has:
\[ 
\INf(S) = \bigcup_{i=1}^{l}\left( \bigcap_{j=1}^{m_i}~\INf(p_{ij} < 0)~\cap~\bigcap_{j=m_i+1}^{M_i}~\INf(p_{ij} = 0) \right)\,.
\]
A proof of this property for semi-algebraic sets~\cite[Lemma 20]{DBLP:conf/emsoft/LiuZZ11} was generalized to semi-analytic sets in~\cite[\S 6.1.2]{DBLP:journals/jacm/PlatzerT20}. 
These results in particular mean that, if one restricts attention to these classes of sets, the equality $\INf(S)^c = \INf(S^c)$ holds (making the equivalence of Theorem~\ref{thm:lzz} and~\cite[Thm. 9.1]{DBLP:books/sp/17/ZWZ2017} immediate, contrary to the general setting where this equality does not hold and where Lemma~\ref{lem:subtlelink} is required to prove the equivalence). 

\subsection{The \LZZ~Decision Procedure Based on Theorem~\ref{thm:charbyinduction}}
\label{sec:lzzalg}
Given a quantifier-free formula describing a semi-algebraic set 
\[ 
S \equiv \bigvee_{i=1}^{l}\left( \bigwedge_{j=1}^{m_i}~p_{ij} < 0~\land~\bigwedge_{j=m_i+1}^{M_i}~p_{ij} = 0 \right)\,,
\]
and a polynomial system of ODEs $x' = f(x)$, in order to decide whether $S$ is a positively invariant set, a basic decision procedure using the characterizations based on inward sets (Theorem~\ref{thm:charbyinduction} and \cite[Thm. 19]{DBLP:conf/emsoft/LiuZZ11}), which we term \LZZ,~after the authors in~\cite{DBLP:conf/emsoft/LiuZZ11}, can be implemented by performing the following steps: 

\begin{enumerate}
\item Compute $\INf(p_{ij} \bowtie_{ij} 0)$, where $\bowtie_{ij} \in \{=,<\}$ appearing in $S$ (formulas $p<0$ and $p=0$, where $p$ is a polynomial, will be referred to as \emph{atomic formulas}), and from these construct 
\[ 
\INf(S) \equiv \bigvee_{i=1}^{l}\left( \bigwedge_{j=1}^{m_i}~\INf(p_{ij} < 0)~\land~\bigwedge_{j=m_i+1}^{M_i}~\INf(p_{ij} = 0) \right)\, ,
\]
following the distributive property of $\INf$ for semi-algebraic sets $S$. 
\item Construct $\INnf(S^c)$ following the same process as in \textsf{step 1}, but using the complement $S^c$ and the reversed system ${x}' = -f({x})$.
\item Check the semi-algebraic set inclusions $S\subseteq \INf(S)$ and $S^c \subseteq \INnf(S^c)$ from Theorem~\ref{thm:charbyinduction} using e.g. the CAD algorithm of~\cite{DBLP:journals/jsc/CollinsH91}.
\end{enumerate}
\begin{remark}
One can alternatively construct $\INnf(S)$ in \textsf{step 2} and check the inclusions $S\subseteq \INf(S)$ and $S^c \subseteq \INnf(S)^c$ in \textsf{step 3}, following the original method of~\cite{DBLP:conf/emsoft/LiuZZ11}, rather than the characterization in Theorem~\ref{thm:lzz}.
\end{remark}
A basic implementation of the \LZZ~decision procedure thus requires an algorithm for computing Gr\"obner bases (to compute the inward sets in \textsf{step 1} and \textsf{step 2}) and a decision procedure for the universally (or existentially) quantified fragment of real arithmetic (to check the semi-algebraic set inclusions in \textsf{step 3}).

In practice, the syntactic description of $S$ may feature atomic formulas that are not of the form $p<0$ or $p=0$, e.g. $S$ may feature the comparison operators $>,\geq,\leq$, and may have atomic formulas where the term on the right-hand side of the comparison operator is not $0$ as assumed above. To implement \textsf{step 1} and \textsf{step 2} for this more general case (without tampering with the description of $S$) it is convenient to compute $\INf(S)$ by syntactically replacing all atomic formulas $p_{\mathrm{lhs}} \bowtie p_{\mathrm{rhs}}$ (where $p_{\mathrm{lhs}}$ and $p_{\mathrm{rhs}}$ are polynomials and $\bowtie\,\in \{ <, \leq, =, \neq, \geq, > \}$) appearing in the syntactic description of $S$, with $\INf(p_{\mathrm{lhs}} \bowtie p_{\mathrm{rhs}})$, which can be defined for atoms in terms of the primitives $\INf(p < 0)$ and $\INf(p = 0)$ in the following way (we use `$:=$' to denote function definitions):
\begin{flalign*}
   \qquad & \INf(\T) := \T\,,& \\
   \qquad & \INf(\F) := \F\,,\\
   \qquad & \INf(p_\mathrm{lhs}=p_\mathrm{rhs}) := \INf(p_\mathrm{lhs}-p_\mathrm{rhs}=0)\,,\\
   \qquad & \INf(p_\mathrm{lhs}<p_\mathrm{rhs}) := \INf(p_\mathrm{lhs}-p_\mathrm{rhs}<0)\,,\\
   \qquad & \INf(p_\mathrm{lhs}>p_\mathrm{rhs}) := \INf(p_\mathrm{rhs}-p_\mathrm{lhs}<0)\,,
\end{flalign*}
and, using the fact that $\INf(S^c) = \INf(S)^c$ for semi-algebraic sets $S$,
\begin{flalign*}
   \qquad & \INf(p_\mathrm{lhs} \neq p_\mathrm{rhs}) := \lnot\, \INf(p_\mathrm{lhs}-p_\mathrm{rhs}=0)\,, \\
   \qquad & \INf(p_\mathrm{lhs}\leq p_\mathrm{rhs}) := \lnot\, \INf(p_\mathrm{rhs}-p_\mathrm{lhs}<0)\,,\\
   \qquad & \INf(p_\mathrm{lhs}\geq p_\mathrm{rhs}) := \lnot\, \INf(p_\mathrm{lhs} - p_\mathrm{rhs}<0)\,.&
\end{flalign*}
The primitives $\INf(p = 0)$ and $\INf(p < 0)$ are defined following Lemma~\ref{lem:INfrem} as
\begin{flalign*}
 \quad \INf(p = 0)  & \quad := \quad (\rem_0=0 \land \rem_1=0 \land \rem_2=0 \land \dots \land \rem_{\rk(p)}=0)\,, \\
 \quad \INf(p < 0)  & \quad := \quad \Big( \rem_0<0  \\
 \quad & \qquad \lor~(\rem_0=0 \land \rem_1<0) \\
 \quad & \qquad \lor~(\rem_0=0 \land \rem_1=0 \land \rem_2<0) \\
 \quad & \qquad  \ \vdots \\
 \quad & \qquad \lor~(\rem_0=0 \land \rem_1=0 \land \rem_2=0 \land \dots \land \rem_{\rk(p)}<0) \Big)\,.
\end{flalign*}

An implementation of the \LZZ~decision procedure in the Wolfram Language can be achieved with fewer than 35 lines of code following the above approach.\footnote{Our implementation is available from~\citep{code}} 
\section{Characterizing Positive Invariance Through Exit Sets}
\label{sec:topology}

In this section we develop an alternative characterization of positively invariant sets using the concept of \emph{exit set} 
as formulated by~\cite{conley}.
Let $s \in I_{x}$ be a point in time within the maximal interval of existence of solution $\varphi$ from initial value $x$, and let $I_{\varphi(s,x)} \deff \{t \mid t+s \in I_{x} \}$, which is simply the time interval $I_{x}$ offset by $s$ (or, equivalently, the maximal interval of existence from the initial value $\varphi(s,x)$). 
The mapping $\varphi$ defines a local flow on the topological space $\bbr^n$ since, for all $x \in \bbr^n$, $\varphi(0,x) = x$, and 
\[
\forall~s \in I_x.~\forall~t \in I_{\varphi(s,x)}. \quad
\varphi(t,\varphi(s,x)) = \varphi(s+t, x)\,.
\] 

Let $S$ be a subset of $\bbr^n$. 
Recall that a point $x \in \bbr^n$ is a closure point of $S$ if and only if every open set containing $x$ intersects $S$ in at least one point (not necessarily distinct from $x$ itself if $x$ happens to be in $S$). Let $S^\circ$ denote the interior of $S$. The boundary of $S$, denoted $\partial S$, is defined as $S \setminus S^\circ$.
As before, we use $t>0$ as a shorthand for $t \in I_x \cap (0,+\infty)$ and, similarly, by $t<0$ we understand $t \in I_x \cap (-\infty,0)$. 
\begin{definition}[Exit Set~\citep{conley}]
\label{def:exitset}
The \emph{exit set} of $S \subseteq \mathbb{R}^n$ with respect to the local flow induced by $x'=f(x)$ is defined as follows:
\begin{equation*}
\exit_f(S) \deff \{ x \in S \mid \forall~t > 0.~\exists~s \in (0,t). ~ \varphi(s,x) \not\in S \}\,.
\end{equation*}
\end{definition}
The exit set of $S$ defines the set of points in $S$ from which the flow cannot evolve forward in time without leaving the set $S$. %
As the name suggests, a flow starting at a point in $\exit_f(S)$ ``leaves the set $S$ immediately'' (regardless of where it was before). 
It is intuitive that such points can only lie on the boundary of $S$.  
\begin{lemma}\label{lem:exitboundary}
The set $\exit_f(S)$ is a subset of $\partial S$ (in addition to being a subset of $S$, by definition). 
\end{lemma}
\begin{proof}
Let $x \in \exit_f(S) \cap S^\circ$, then there exists an open set $U \subset S^\circ$ containing $x$. By continuity of $\varphi(\cdot,x)$ with respect to time, there exists a neighbourhood $I$ of $0$ in $I_x$ such that $\varphi(t,x) \in U$ for all $t \in I$. 
Let $t \in I \cap (0,+\infty)$. Since $x\in \exit_f(S)$, there exists $s \in (0,t) \subset I$ such that $\varphi(s,x) \not\in S$ and, a fortiori, $\varphi(s,x) \not\in U$, which contradicts the existence of $I$ and thus $\exit_f(S) \cap S^\circ = \emptyset$. Since $\exit_f(S) \subseteq S$ by definition, the exit set is a subset of $\partial S$.
\end{proof}

Positive invariance of a set $S$ (as given in Definition~\ref{def:posinv}) may be equivalently defined using the set of  so-called \emph{escape points}~\citep[also due to][]{conley}:~\footnote{The set of escape points is fundamental to the Wa\.zewski principle. See~\cite{conley} where it is denoted as $W^\circ$ for a set $W$.} 
\begin{equation}
    \label{eq:escape}
    \escape_f(S) \deff \{ x \in S \mid \exists~t>0.~\varphi(t,x) \not\in S\}\,.
\end{equation}

Notice the difference between the exit and escape sets: starting at an exit point, the flow \emph{immediately exits} the set $S$, whereas for an escape point the flow may first evolve within $S$ before leaving $S$ at some point in time in the future (i.e., it \emph{must} eventually leave $S$). 
Thus, $\exit_f(S) \subseteq \escape_f(S)$. 
The set of escape points of $S$ is empty precisely when $S$ is a positively invariant set. Furthermore, this criterion can be stated entirely in terms of exit sets. 
\begin{theorem}
\label{thm:char}
A set $S\subseteq \mathbb{R}^n$ is positively invariant if and only if both $\exit_f(S)$ and $\exit_{-f}(S^c)$ are empty. 
\end{theorem}
\begin{proof}
For necessity, it is easy to see that the set is not positively invariant whenever the exit sets are not both empty. Case (i): if $\exit_f(S)$ is non-empty, then for some point $x \in S$ there exists a $t>0$ such that $\varphi(t,x) \not\in S$. Case (ii): if $\exit_{-f}(S^c)$ is non-empty, then for some $y \not\in S$ there exists a $\tau>0$ such that $\varphi(-\tau, y) \in S$. Taking $z = \varphi(-\tau, y)$, it is clear that $z \in S$ and $\varphi(\tau, z) \not\in S$. 

For sufficiency we show that whenever $S$ is not positively invariant, the sets $\exit_f(S)$ and $\exit_{-f}(S^c)$ cannot both be empty.  Suppose (for contradiction) that both $\exit_f(S)$ and $\exit_{-f}(S^c)$ are empty and that $S$ is not positively invariant. The set of escape points of $S$ is therefore non-empty. Consider an escape point $x \in \escape_f(S)$: by our hypothesis $x$ cannot be in the empty set of exit points $\exit_{f}(S)$. Therefore there exists a positive $t_0 \in I_x$ such that for all $s \in (0,t_0)$ one has $\varphi(s,x)\in S$, and there exists a $t_1 \in I_x$ such that $t_0 \leq t_1$ and $\varphi(t_1,x) \not\in S$~(i.e. $\varphi(t_1,x) \in S^c$). Let us define
\[
T' = \{t \in I_x \cap (0,+\infty) \mid \forall s \in (0,t), \varphi(s,x) \in S \}\,.
\]
Under our hypothesis, the set $T'$ is non-empty and has a supremum $t'$ such that $t_0\leq t'\leq t_1$. Let us now define 
\[
T'' = \{t \in I_x \cap (0,+\infty) \mid \varphi(t,x) \not\in S\}\,.
\]
This set is likewise non-empty (as it contains $t_1$) and has an infimum $t''$ such that $t_0 \leq t''$. Every element of $T''$ is an upper bound for $T'$ (otherwise there would exist a time $t\in T''$ at which both $\varphi(t,x) \in S$ and $\varphi(t,x) \not\in S$). Clearly, since $t'$ is the least upper bound for $T'$ it can act as a lower bound on $T''$ and we therefore have $t' \leq t''$, where $t''$ is the greatest lower bound for $T''$. Suppose the inequality is strict ($t'<t''$), then for all $r \in [t',t'')$ one has $\varphi(r,x) \in S$ (otherwise $t''$ is not the greatest lower bound for $T''$). But then $t'$ cannot be the least upper bound for $T'$ because $\varphi(s,x)\in S$ for $s \in (t',t'')$. Thus $t'=t''$ and we have two cases to consider: either (i) $\varphi(t',x)\in S$, in which case $\varphi(t',x)\in \exit_f(S)$ and $\exit_f(S)$ is therefore non-empty, or (ii) $\varphi(t',x) \not\in S$, in which case $\varphi(t',x) \in \exit_{-f}(S^c)$, so $\exit_{-f}(S^c)$ is non-empty. Both cases give us a contradiction.
\end{proof}
\begin{remark}
The main technical difference between the proof of Theorem~\ref{thm:char} and~\cite[Thm. 9.1]{DBLP:books/sp/17/ZWZ2017} is that the latter draws a contradiction from considering the supremum of the set $T'$ (with respect to our notations in the proof of Theorem~\ref{thm:char}) whereas we draw a contradiction by considering in addition the set $T''$. This is to be expected as the statements of these theorems are slightly different: Theorem~\ref{thm:char} complements $S$ then applies the $\exit_{-f}$ operator to $S^c$, whereas~\cite[Thm. 9.1]{DBLP:books/sp/17/ZWZ2017} applies the $\INnf$ operator to $S$ first then complements the result. 
This being said, the overall structure of both proofs is however very similar and this fact is better captured by appealing to the real induction principle as a generic proof technique as detailed in Section~\ref{sec:realinduction}. 
\end{remark}

\begin{remark}
Readers with a background in dynamical systems may find it a little counterintuitive that one needs to consider the flow in the reversed system to characterize positive invariance. Indeed, for \emph{closed sets} $S$ it is well known that ``local invariance'' under the flow $\varphi$ (viz. emptiness of $\exit_f(S)$) is equivalent to positive invariance~\citep[e.g. see][Ch. 4]{vrabie2007}. It is important to remember that Theorem~\ref{thm:char} makes no assumptions about the set $S$ being open or closed. When $S$ is open, local invariance holds trivially because the flow may always evolve within the set for some time from any $x\in S$. 

\end{remark}

Observe that the sets $\exit_f(S)$ and $\exit_{-f}(S)$ are not necessarily disjoint: for example, any isolated point which is not an equilibrium would lie in both sets. 
Neither are they required to cover the boundary $\partial S$: if $S$ is an equilibrium point, then both $\exit_f(S)$ and $\exit_{-f}(S)$ are empty, whereas $\partial S = S$. 
The operators $\exit_f$ and $\INf$ respectively capturing the main underlying concepts used in Theorems~\ref{thm:charbyinduction} and Theorem~\ref{thm:char} are intimately related. 
\begin{lemma}
\label{lem:link}
For any set $S \subseteq \bbr^n$, $\exit_f(S) = \INf(S)^c \cap S$. 
Equivalently, one has $\exit_f(S)^c \cap S = \INf(S) \cap S$. 
\end{lemma}
\begin{proof}
One has $x \in \INf(S)^c \cap S$ if and only if $x \in S$, and, for any positive $t \in I_x$, there exists $s \in (0,t)$ such that $\varphi(s,x) \not\in S$, otherwise $\varphi(s,x) \in S$ holds for all $s \in (0,t)$ which would mean that $x \in \INf(S)$. 
The latter is exactly the definition of $\exit_f(S)$. 
\end{proof}
A symmetric equality holds (only) for closed sets. 
\begin{lemma}
\label{lem:linkclosed}
For a \emph{closed} set $S \subseteq \bbr^n$, \mbox{$\INf(S) = \exit_f(S)^c \cap S$}.
\end{lemma}
\begin{proof}
If $S$ is a closed set, then the inclusion $\INf(S) \subseteq S$ holds trivially, from Lemma~\ref{lem:link} we have $\exit_f(S)^c \cap S = \INf(S) \cap S$ and the result follows. 
\end{proof}
\begin{remark}
According to the above lemmas, while $\INf(S)$ is sufficient to fully recover $\exit_f(S)$ by simple set operations. The converse is not true for general sets: the bare knowledge of $\exit_f(S)$ is not enough to completely recover $\INf(S)$ unless $S$ is closed. 
This might seem as a conceptual defect favouring inward sets as more fundamental than exit sets. From a computational standpoint, however, this lack of symmetry between the two concepts turns out to be powerful: intuitively one does not need the full information encoded by inward sets to decide the positive invariance of $S$. Exit sets, despite carrying less information, are sufficient for this purpose. 
\end{remark}
Using Lemma~\ref{lem:link}, both characterizations of positively invariant sets in Theorem~\ref{thm:charbyinduction} and Theorem~\ref{thm:char} can be recovered from one another using the following equivalences:
\[
\emptyset = \underbrace{\INf(S)^c \cap S}_{\exit_f(S)} \iff S \subseteq \INf(S)\,,
\]
\[
\emptyset = \underbrace{\INnf(S^c)^c \cap S^c}_{\exit_{-f}(S^c)} \iff S^c \subseteq \INnf(S^c)\,.
\]

The origins of exit sets in Theorem~\ref{thm:char} lie in topology and it is the properties of exit sets that make this characterization computationally interesting. 
The astute reader may remark at this point that Theorem~\ref{thm:char} admits a shorter proof using real induction via Lemma~\ref{lem:link}. This is indeed the case; however, such a proof would not rely on the concept of exit set nor would it expose the topological insights that we wish to call upon later. 
As we shall see, exit sets afford a very different way of looking at the problem of checking positive invariance and their properties can be exploited to give a substantially different algorithmic solution than that offered by~\LZZ{} in Section~\ref{sec:lzzalg}.

\subsection{Properties of Exit Sets}
\label{subsec:exitproperties}
Let $S_1, S_2 \subseteq \bbr^n$, we discuss below the distributive properties of $\exit_f$ over set intersection and union. 
\begin{lemma}\label{lem:exitcap}
\(
\exit_f(S_1 \cap S_2) = (\exit_f(S_1) \cap S_2) \cup (S_1 \cap \exit_f(S_2))
\). 
\end{lemma}
\begin{proof}
The inclusion $\exit_f(S_1 \cap S_2) \supseteq (\exit_f(S_1) \cap S_2) \cup ( \exit_f(S_2) \cap S_1)$ is immediate: if $x \in \exit_f(S_1) \cap S_2$, then, for all positive $t$, there exists a positive $s < t$ such that $\varphi(s,x) \not\in S_1$ and therefore $\varphi(s,x) \not\in S_1 \cap S_2$. 
Likewise for $\exit_f(S_1 \cap S_2) \supseteq \exit_f(S_2) \cap S_1$. 
To prove the converse, let $x \in \exit_f(S_1 \cap S_2)$, then $x \in S_1 \cap S_2$ and for all positive $t$, there exists a positive $s < t$ such that $\varphi(s,x) \not\in S_1 \cap S_2$ which is equivalent to $\varphi(s,x) \not\in S_1$ or $\varphi(s,x) \not\in S_2$. %
\end{proof}

\begin{lemma}\label{lem:exitcup} 
\( 
\exit_f(S_1 \cup S_2) \subseteq \bigl(\exit_f(S_1) \cap \INf(S_2)^c\bigr) \cup \bigl( \INf(S_1)^c \cap \exit_f(S_2) \bigr) 
\). 
\end{lemma}
\begin{proof}
Let $x \in \exit_f(S_1 \cup S_2)$, then by definition, for all $t >0$, there exists $s \in (0,t)$ such that $\varphi(s,x) \not\in S_1 \cup S_2$, which is equivalent to $\varphi(s,x) \not\in S_1$ and $\varphi(s,x) \not\in S_2$. 
By hypothesis, $x \in S_1 \cup S_2$. If $x \in S_1$ then it has to belong to $\exit_f(S_1)$ as well as $\INf(S_2)^c$, by definition of the latter.
If $x \in S_2$, we get a symmetric formula by swapping $S_1$ and $S_2$, namely $x \in \exit_f(S_2) \cap \INf(S_1)^c$. 
The desired formula is the union of these two cases. 
\end{proof}

\begin{counterexample}
The reverse inclusion of Lemma~\ref{lem:exitcup} does not hold in general. 
Consider the simple $1$-dimensional system $x' = 1$ and the sets    
\begin{align*}
S_1 &= \{0\} \cup \left\{ x \in \bbr \mid x > 0 \land \sin\left(x^{-1}\right)=0  \right\}\,, \\
S_2 &= \{0\} \cup \left\{ x \in \bbr \mid x > 0 \land \sin\left(x^{-1}\right)\neq 0 \right\}\,.
\end{align*}
The point $0$ belongs to both $\exit_f(S_1)$ and $\exit_f(S_2)$. In addition, it does not belong to either $\INf(S_1)$ or $\INf(S_2)$. However, $0$ is not in $\exit_f(S_1 \cup S_2)$ as the union ($x \geq 0$) is clearly a positively invariant set for the considered flow. 
\end{counterexample}
This simple example highlights the main reason why the inclusion in Lemma~\ref{lem:exitcup} cannot in general be replaced with set equality. If $x \in \exit_f(S_1) \cap \INf(S_2)^c$, one can only conclude that for any positive $\epsilon_1,\epsilon_2$, there exist $t_1 \in (0,\epsilon_1)$ and $t_2 \in (0,\epsilon_2)$ such that $\varphi(t_1,x) \not\in S_1$ and $\varphi(t_2,x) \not\in S_2$; there is nothing to suggest that $t_1$ should be equal to $t_2$, which is required for $x$ to belong to $\exit_f(S_1 \cup S_2)$. 
However, if one restricts attention to semi-analytic sets $S$ (which includes semi-algebraic sets) then $\INf(S^c) = \INf(S)^c$ (as observed in the previous section), and the inclusion of Lemma~\ref{lem:exitcup} becomes an equality. 
(Notice that semi-analyticity is only sufficient to ensure that $\INf(S^c) = \INf(S)^c$. We currently lack a full characterization of the most general topological settings that respect this equality.)  
 
\begin{lemma}
\label{lem:exitcupSA}
Let $S_1, S_2$ be semi-analytic sets. Then  
\[ 
\exit_f(S_1 \cup S_2) = \bigl(\exit_f(S_1) \cap \INf(S_2)^c\bigr) \cup \bigl(\INf(S_1)^c \cap \exit_f(S_2) \bigr)\,.
\]
\end{lemma}
\begin{proof}
\begin{align*}
    \exit_f(S_1 \cup S_2) 
    &= \INf(S_1 \cup S_2)^c \cap (S_1 \cup S_2) \\
    &= (\INf(S_1)^c \cap \INf(S_2)^c \cap S_1) \cup   (\INf(S_1)^c \cap \INf(S_2)^c \cap S_2) \\
    &= \bigl(\exit_f(S_1) \cap \INf(S_2)^c\bigr) \cup \bigl(\exit_f(S_2) \cap \INf(S_1)^c\bigr)\,.
\end{align*}
\end{proof}

\subsection{The \ExitSet{} Decision Procedure Based on Theorem~\ref{thm:char}}
\label{sec:exitsetalg}

Given a quantifier-free formula describing a semi-algebraic set $S$ and a polynomial system of ODEs $x'=f(x)$, Theorem~\ref{thm:char} can be used to algorithmically decide whether $S$ is positively invariant or not with respect to $f$. 

A na\"ive approach would be to first compute $E = \exit_f(S) \cup \exit_{-f}(S^c)$ recursively on the Boolean structures of $S$ and $S^c$ using Lemmas~\ref{lem:exitcup} and~\ref{lem:exitcupSA}, then check whether $E$ is empty or not. 
Such an approach would be very similar to the \LZZ{} procedure described in  section~\ref{sec:lzzalg} %
and would therefore suffer from the same problem, namely the impossibility of the current state-of-the-art quantifier elimination algorithms to check the emptiness of $E$ in reasonable time, even for seemingly simple planar systems (cf. section~\ref{sec:ex}). 
Indeed, one experimentally observes that, for many interesting examples, the construction of the set $E$ is not computationally expensive despite requiring several ideal membership tests as, often, the order (with respect to $f$) of the polynomials involved remains low. 
An overwhelming share of the running time for a typical problem is spent on proving emptiness of $E$ (as is the case for checking the inclusions $S \subseteq \INf(S)$ and $S^c \subseteq \INnf(S^c)$ using \LZZ). 

We will see in this section how the concept of exit sets, and more precisely Theorem~\ref{thm:char}, can be used to overcome this bottleneck in a principled way. 
The main idea is to ``chop the set $E$ into smaller pieces'' (chunks) on which the emptiness test can be performed in a divide-and-conquer fashion, instead of constructing a formula characterizing $E$ first and only then checking for its emptiness. 
What is perhaps more interesting is that Theorem~\ref{thm:char} suggests a natural way of splitting $E$ into chunks in such a way that each chunk involves precisely \emph{one} exit set of an atomic formula. 
This, in turn, allows one to exploit topological properties of atomic formulas, such as openness, in order to check for set emptiness \emph{syntactically}, obviating the need for expensive computations such as real quantifier elimination. 

As in the previous sections, we use the same notation for semi-algebraic sets and their formal representations as quantifier-free formulas of real arithmetic. 
Without loss of generality, we also restrict our attention to the atomic formulas, $p < 0$ and $p=0$, where $p$ is a polynomial. 
The formulas $p \leq 0$ and $p \neq 0$, are syntactic sugar for $(p < 0 \lor p=0)$ and $(-p<0 \lor p<0)$ respectively. %
Similarly, $p>0, p \geq 0$ can be encoded as $-p < 0, -p \leq 0$ respectively.  

The exit sets of $\F$, $\T$, and $p<0$  are all empty (by Lemma~\ref{lem:exitboundary}) as the sets defined by these formulas are open. 
According to the same lemma, the exit set of $p = 0$ necessarily lies on its boundary, which is also given by $p=0$. When the first Lie derivative of $p$ does not vanish on $p=0$, the flow necessarily leaves the set for some positive time. The same reasoning applies for higher-order Lie derivatives. As with the construction of $\INf$ in Section~\ref{sec:lzzalg}, the construction of the exit set of $p=0$ is fully captured by a (finite) formula: 
\begin{align*}
    \exit_f(p=0) &\equiv \big(~ p=0 \land p' \neq 0 \\
  &\quad \lor p=0 \land p'=0 \land p''\neq 0 \\
  &\quad \vdots& \\
 &\quad \lor p=0 \land p'=0 \land p''=0 \land \cdots \land p^{(\rk(p))}\neq0 \big)\,.
\end{align*}
Note that Lemma~\ref{lem:rem} also applies to $\exit_f(p = 0)$ and the remainders $\rem_i$ (as defined in the lemma) can be used instead of the Lie derivatives $p^{(i)}$. 
In summary, the exit set of atomic formulas can be constructed using a procedure $\exit_f$ which is defined as follows:  
\begin{flalign*}
\quad & \quad \exit_f(\F) := \F\,,\\
\quad & \quad \exit_f(\T) := \F\,,\\
\quad & \quad \exit_f(p<0) := \F\,,\\
\quad & \quad \exit_f(p=0) := \big(~ \rem_0=0 \land \rem_1 \neq 0 \\
  \quad & \phantom{\quad \quad \exit_f(p=0) :=} \lor \rem_0=0 \land \rem_1=0 \land \rem_2\neq 0 \\
  \quad & \phantom{\quad \quad \exit_f(p=0) :=} \vdots \\
 \quad & \phantom{\quad \quad \exit_f(p=0) :=} \lor \rem_0=0 \land \rem_1=0 \land \cdots \land \rem_{\rk(p)}\neq0 \big)\,.
\end{flalign*}
Thus, the only non-trivial exit set for atomic formulas is the exit set of an equality as it is the only (atomic) closed set. 

We next define a recursive procedure called $\IsEmpty_f$, parametrized by the vector field $f$, and which takes as its arguments two quantifier-free real arithmetic formulas describing semi-algebraic sets $S$ and $R$. 
It is defined as follows ($A$ denotes an atomic formula): %
\begin{flalign*}
  \qquad  & \IsEmpty_f(A,~R) := \Reduce\left( \exists x_1. \dotsc \exists x_n.~\exit_f(A) \land R \right)\,, & \\
   &\IsEmpty_f(S_1 \land S_2,~R) := 
   \IsEmpty_f(S_1,~S_2 \land R) \\ 
  &\phantom{\IsEmpty_f(S_1 \land S_2,~R) :=} \lor \IsEmpty_f(S_2,~S_1 \land R)\,,\\
   &\IsEmpty_f(S_1 \lor S_2,~R) := \IsEmpty_f(S_1, \lnot\INf(S_2) \land R) \\
   &\phantom{\IsEmpty_f(S_1 \lor S_2,~R) :=} \lor \IsEmpty_f(S_2,~ \lnot\INf(S_1) \land R)\,, \\
    &\IsEmpty_f(\lnot S,~R) := \IsEmpty_f(\Neg(S),~R) \, .
\end{flalign*}
In addition to $\exit_f$, $\IsEmpty_f{}$ relies on three other procedures: $\INf$ (already defined in Section~\ref{sec:lzzalg}), %
$\Neg{}$, and $\Reduce{}$. 
The procedure $\Neg{}$ applies negation $\lnot$ to the formula it receives as its argument (but does \emph{not} recursively apply negation to the sub-formulas). For atomic formulas, $\Neg{}$ simply negates the formula, expressing the result in terms of only the basic forms of atomic formulas ($\T$, $\F$, $p<0$, and $p=0$):
\begin{flalign*}
    \qquad & \Neg(\F) := \T\,, &\\
    & \Neg(\T) := \F\,,\\
    & \Neg(p<0) := (-p<0) \lor (-p=0)\,, \\
    & \Neg(p=0) := (-p<0) \lor (p<0)\,. 
\end{flalign*}
For non-atomic formulas $\Neg{}$ simply applies De Morgan's laws and eliminates double negation: %
\begin{flalign*}
    \qquad & \Neg(S_1 \land S_2) := (\lnot S_1) \lor (\lnot S_2)\,, & \\ 
    & \Neg(S_1 \lor S_2) := (\lnot S_1) \land (\lnot S_2) \,, \\
    & \Neg(\lnot S) := S\,.
\end{flalign*}

The procedure $\Reduce$ checks for emptiness of the semi-algebraic set by performing real quantifier elimination (this functionality is offered e.g. by implementations of CAD~\citep{DBLP:journals/jsc/CollinsH91}; however, there exist alternatives which are not based on CAD, e.g. RAGLib~\citep{raglib}).

In a nutshell, the main purpose of $\IsEmpty_f$ is to \emph{recursively} check for emptiness of exit sets, as stated more formally in the following lemma. \begin{lemma}
\label{lem:isempty}
Let $S$ and $R$ be two formulas describing semi-algebraic sets. Then $\IsEmpty_f(S,R)$ returns \textsf{False} if and only if $\exit_f(S) \cap R$ is empty.
\end{lemma}
\begin{proof}
The proof is by induction on the depth of formula $S$. 
Base case: if $S \in \{\F, \T, p<0, p=0\}$, then $\IsEmpty_f(S,R)$ is  
\[
\Reduce\left( \exists x_1. \dotsc \exists x_n.~\exit_f(S) \land R \right)\,,
\]
which is \textsf{False} if and only if $\exit_f(S) \cap R$ is empty (we freely interchange $\land$ and $\cap$ as well as the empty set and \textsf{False} as mentioned in  Remark~\ref{rmk:notation}). %

For the inductive hypothesis, suppose the property holds for all formulas of depth less than or equal to $k$ and let $S_1, S_2,$ and $S'$ be such formulas.

If $S = S_1 \land S_2$, then by definition $\IsEmpty_f(S_1 \land S_2,R)$ is \textsf{False} if and only if both $\IsEmpty_f(S_1,~S_2 \land R)$ and $\IsEmpty_f(S_2,~S_1 \land R)$ are \textsf{False}. 
By the induction hypothesis, this means that both $\exit_f(S_1) \land (S_2 \land R)$ and $\exit_f(S_2) \land (S_1 \land R)$ are empty, and therefore their union is also empty. 
One gets the desired result by factoring out $R$ then using Lemma~\ref{lem:exitcap}: 
\begin{align*}
    \emptyset=&~\left( \exit_f(S_1) \cap (S_2 \cap R)\right) \cup \left( \exit_f(S_2) \cap (S_1 \cap R)\right) \\
    =&~ \left( (\exit_f(S_1) \cap S_2) \cup (\exit_f(S_2) \cap S_1) \right) \cap R \\
    =&~\exit_f(S_1\cap S_2) \cap R\,.
\end{align*}
The disjunctive case can be proved similarly using Lemma~\ref{lem:exitcupSA}. 

Finally, if $S = \lnot S'$, 
$\IsEmpty_f(S, R) = \IsEmpty_f(\Neg{}(S'), R)$
and one may eliminate all negations from $\Neg{}(S')$ by applying De Morgan's laws and double negation elimination and finally applying $\Neg{}$ to any remaining negated atoms. Since the property holds for atomic formulas, conjunctions and disjunctions, it holds for $S$ as well. 
\end{proof}

The procedure $\IsEmpty_f{}$ can thus be used to check positive invariance as an immediate corollary of Theorem~\ref{thm:char} and Lemma~\ref{lem:isempty} by setting $R$ to $\T$.
\begin{theorem}\label{thm:isemptychar}
A semi-algebraic set $S$ is positively invariant for a system of ODEs $x'=f(x)$ if and only if
\(  \lnot \left(\IsEmpty_f(S,\T)  ~\lor~ \IsEmpty_{-f}(\lnot S, \T) \right) \).
\end{theorem}
Accordingly, we define the \emph{Exit Set Emptiness} ($\ExitSet{}$) decision procedure that checks for positive invariance of $S$ with respect to $f$ as 
\[
\ExitSet(S,f) := \lnot \left(\IsEmpty_f(S,\T)  ~\lor~ \IsEmpty_{-f}(\lnot S, \T) \right)\,.
\]
\subsection{Complexity Analysis}
\label{sec:complexity}

For given formulas $S$ and $R$, in order to check the emptiness of $\exit_f(S) \cap R$, the procedure $\IsEmpty_f(S,R)$ performs several calls to $\Reduce{}$ in order to eliminate existential quantifiers.
The number of such calls depends only on the Boolean structure of $S$, in particular, the second argument $R$ plays no role in the way the procedure operates. 
In this section, we first give upper and lower bounds of the number of such calls as a measure of the impact of the encoding of $S$.  
We then discuss further decompositions of $\exit_f(S)$ as a union of \emph{basic} semi-algebraic sets. Recall that a basic semi-algebraic set is a set described by a conjunction of atomic formulas $\bigwedge_i (p_i \bowtie_i 0)$, where $\bowtie_i \in \{<,=\}$ and $p_i$ are polynomials. 

\begin{proposition}
\label{prop:DNFsizes}
Suppose the set $S$ is characterized by a formula in disjunctive normal form (DNF) $\bigvee_{i=1}^k \bigwedge_{j=1}^{m_i} A_{i j}$, where $A_{i j}$ are atomic formulas. Let $m = \max_i m_i$. 
Then the recursion depth of $\IsEmpty_f(S,\T)$ is bounded by $k + m$ and the number of calls to $\Reduce{}$ is $\sum_{i=1}^k m_i \leq k m$, each of which has the form $\Reduce\, \exists x_1 \dotsc \exists x_n. \exit_f(A_{r s}) \land R_{r s}$, where 
\[
R_{r s} \equiv \bigwedge_{j=1,j\neq s}^{m_r} A_{r j} \land \lnot \INf\left( \bigvee_{i=1,i \neq r}^k \bigwedge_{j=1}^{m_i} A_{i j} \right)\,.
\]
\end{proposition}
\begin{proof}
The form of the real quantifier elimination (QE) problems is immediate from the definition of $\IsEmpty_f$. The equivalence of $R_{r s}$ is obtained by using the distributive properties (over disjunctions and conjunctions) of the $\INf$ operator. 
\end{proof}

For instance, suppose $S \equiv (A_{1 1} \land A_{1 2}) \lor A_{2 1} \lor A_{3 1}$ ($k=3$, $m=m_1=2$, $m_2=m_3=1$). Then, in the worst case, the procedure $\IsEmpty_f(S,\T)$ has to call $\Reduce{}$ $4$ times: 
\begin{align*}
&\Reduce~ \exists x_1 \dotsc \exists x_n.~ \exit_f(A_{1 1}) \land A_{1 2} \land \lnot \INf(A_{2 1} \lor A_{3 1})\,, \\ 
            &\Reduce~ \exists x_1 \dotsc \exists x_n.~ \exit_f(A_{1 2}) \land A_{1 1} \land \lnot \INf(A_{2 1} \lor A_{3 1})\,, \\
            &\Reduce~ \exists x_1 \dotsc \exists x_n.~  \exit_f(A_{2 1}) \land \lnot \INf((A_{1 1} \land A_{1 2}) \lor A_{3 1})\,, \\
            &\Reduce~ \exists x_1 \dotsc \exists x_n.~ \exit_f(A_{3 1}) \land \lnot\INf((A_{1 1} \land A_{1 2}) \lor A_{2 1})\,.
\end{align*}

\begin{proposition}
\label{prop:CNFsizes}
Suppose the set $S$ is characterized by a formula in conjunctive normal form (CNF) $\bigwedge_{i=1}^k \bigvee_{j=1}^{m_i} A_{i j}$ where $A_{i j}$ are atomic formulas, and let $m = \max_i m_i$. 
Then the recursion depth of $\IsEmpty_f(S,\T)$ is bounded by $k + m$ and the number of calls to $\Reduce{}$ is $\sum_{i=1}^k m_i \leq k m$, each of which has the form $\Reduce \, \exists x_1 \dotsc \exists x_n. \exit_f(A_{r s}) \land R_{r s}$, where 
\[
R_{r s} \equiv \lnot\INf\left( \bigwedge_{j=1,j\neq s}^{m_r} A_{r j} \right) \land  \bigwedge_{i=1,i \neq r}^k \bigvee_{j=1}^{m_i} A_{i j}\,.
\]
\end{proposition}

For instance, suppose $S \equiv (A_{1 1} \lor A_{1 2}) \land A_{2 1} \land A_{3 1}$, ($k=3$, $m=m_1=2$, $m_2=m_3=1$). Then, in the worst case, the procedure $\IsEmpty_f(S,\T)$ has to call $\Reduce{}$ $4$ times: 
\begin{align*}
&\Reduce~ \exists x_1 \dotsc \exists x_n.~ \exit_f(A_{1 1}) \land \lnot \INf(A_{1 2}) \land (A_{2 1} \land A_{3 1})\,, \\ 
            &\Reduce~ \exists x_1 \dotsc \exists x_n.~ \exit_f(A_{1 2}) \land \lnot\INf(A_{1 1}) \land (A_{2 1} \land A_{3 1})\,, \\
            &\Reduce~ \exists x_1 \dotsc \exists x_n.~  \exit_f(A_{2 1}) \land ((A_{1 1} \land A_{1 2}) \lor A_{3 1})\,, \\
            &\Reduce~ \exists x_1 \dotsc \exists x_n.~ \exit_f(A_{3 1}) \land ((A_{1 1} \land A_{1 2}) \lor A_{2 1})\,.
\end{align*}

\begin{remark}
    Suppose $S \equiv \bigvee_{i=1}^k \bigwedge_{j=1}^{m_i} A_{i j}$ and let $S'$ denote the same formal expression as $S$ except that $\lor$ and $\land$ are swapped. Then the QE problems that $\IsEmpty_f(S',\T)$ has to solve could be obtained syntactically from those of $\IsEmpty_f(S,\T)$ by swapping $A_{i j}$ and $\lnot\INf(A_{i j})$ (and leaving $\exit_f(A_{i j})$ untouched). 
\end{remark}

The encoding of the set $S$ to be checked may have a significant impact on the number of calls to $\Reduce{}$ in $\IsEmpty_f(S,\T)$. 
For instance, suppose $S$ is encoded as $S_1 \equiv (A_1 \lor (A_2 \land A_3)) \land (A_4 \lor (A_2 \land A_3))$ where the $A_i$ are atomic formulas. Then $\IsEmpty_f(S_1,\T)$ calls $\Reduce{}$ $6$ times. In this case, none of the upper bounds of Propositions~\ref{prop:DNFsizes} nor~\ref{prop:CNFsizes} apply because $S_1$ is neither in DNF nor in CNF. 
If one uses the equivalent (DNF) encoding $S_2 \equiv (A_1 \land A_4) \lor (A_2 \land A_3)$ for $S$, then $\IsEmpty_f(S_2,\T)$ calls $\Reduce{}$ only $4$ times at most. 
\begin{lemma}
\label{lem:lowerb}
The number of calls to $\Reduce{}$ is bounded below by the number of distinct atomic formulas in $S$ (regardless of the encoding of $S$). 
\end{lemma}
\begin{proof}
The procedure $\IsEmpty_f$ requires one call to $\Reduce{}$ for each problem of the form $\exit_f(A) \land R$ (where $A$ is an atomic formula), and $R$ any arbitrary formula. 
Depending on the encoding of $S$, $\IsEmpty_f$ might call $\Reduce{}$ once for $\exit_f(A) \land (R_1 \lor R_2)$, or twice for $\exit_f(A) \land R_1$ and $\exit_f(A) \land R_2$ separately. 
In the best case, the encoding of $S$ is such that each call to $\Reduce{}$ features a distinct $\exit_f(A)$ (otherwise, the several calls with the same $\exit_f(A)$ can be factored out), and the result follows. 
\end{proof}

An interesting open question is whether there exists a systematic way of finding an encoding of $S$ which always results in the minimal number of calls to $\Reduce{}$ that is possible. 
We leave this question open while observing that one can build simple examples for which neither the DNF nor the CNF encoding of $S$ are adequate in this regard (it suffices to consider encodings with redundant atomic formulas). 

The QE problems to solve in Proposition~\ref{prop:DNFsizes} can be split further (by distributivity) into $\prod_{i=1,i\neq r}^k m_i \leq m^{k-1}$ ``smaller'' problems of the form 
\[
\Reduce~ \exists x_1 \dotsc \exists x_n.~ \exit_f(A_{r s}) \land \bigwedge_{j=1,j\neq s}^{m_r} A_{r j} \land \bigwedge_{i=1,i \neq r}^{k} \lnot \INf(A_{i \ell_i})\,.
\]
Likewise, the QE problems to solve in propositions~\ref{prop:CNFsizes} can be split further into $\prod_{i=1,i\neq r}^k m_i \leq m^{k-1}$ problems of the form 
\[
\Reduce~ \exists x_1 \dotsc \exists x_n.~ \exit_f(A_{r s}) \land \bigwedge_{j=1,j\neq s}^{m_r} \lnot \INf(A_{r j}) \land \bigwedge_{i=1,i \neq r}^{k} A_{i \ell_i}\,.
\]
We could further evaluate $\exit_f$ and $\INf$ for atomic formulas. 
To do so, one has to account for the system of ODEs $x'=f(x)$ as well as the order of the involved polynomials with respect to $f$. 
Let $\deg(p)$ denote the (total) degree of a polynomial $p$, and $\deg(f)$ the maximum degree of the polynomials appearing in the right-hand side of $x'=f(x)$. 
Recall that the degree of $p'$, the first (Lie) derivative of $p$ with respect to $f$, has a total degree which is at most $\deg(p)+(\deg(f)-1)$, and  the degree of $p^{(s)}$ is at most $\deg(p)+s(\deg(f)-1)$. 
Recall that $\rk(p)$ denotes the order of $p$ with respect to $f$. 

The set $\exit_f(p \bowtie 0)$ is the union of $\rk(p)$ basic semi-algebraic sets, whereas $\INf(p \bowtie 0)$ is the union of $\rk(p)+1$ basic semi-algebraic sets. 
\begin{lemma}
\label{lem:disjunctDNF}
Let $p_i$, $1 \leq i \leq m$, and $q_j$, $1 \leq j \leq k$, denote some polynomials and let $\rho$ denote the maximum of their respective order with respect to $f$. The expression 
\[
\exit_f(p_1 \bowtie_1 0) \land \bigwedge_{i=2}^{m} (p_i \bowtie_i 0) \land \bigwedge_{j=1}^{k} \INf\bigl( q_j \bowtie_j 0 \bigr)
\]
is the union of at most $\rho(\rho+1)^{k}$ basic semi-algebraic sets. 
Each basic semi-algebraic set is a conjunction of at most $m-1+(k+1)(\rho+1)$ expressions of the form $p \bowtie 0$.
\end{lemma}
\begin{proof}
The expression is a union of at most $\rk(p_1) \prod_{j=1}^{k} (\rk(q_j) + 1)$ basic semi-algebraic sets. From which one immediately deduces the $\rho(\rho+1)^{k}$ upper bound. 
Each basic semi-algebraic set is a conjunction of at most $(\rk(p_1)+1) + (m-1) + \sum_{j=1}^{k} (\rk(q_j)+1) \leq m+k + (k+1) \rho$ literals. 
\end{proof}

For instance, $\exit_f(p_1 = 0) \land (p_2 < 0) \land \INf(q < 0)$ where $\rk(p_1)=\rk(q)=2$, (thus $m=2$, $k=1$, and $\rho=2$) is the following union 
\begin{align*}
    \phantom{~\lor~}     &p_1 = 0 \land p'_1 \neq 0 \land p_2 < 0 \land q < 0  \\
    ~\lor~               &p_1 = 0 \land p'_1 \neq 0 \land p_2 < 0 \land q = 0 \land q' < 0  \\
    ~\lor~               &p_1 = 0 \land p'_1 \neq 0 \land p_2 < 0 \land q = 0 \land q'=0 \land q'' < 0  \\
    ~\lor~               &p_1 = 0 \land p'_1 = 0 \land p''_1 \neq 0 \land p_2 < 0 \land q < 0  \\
    ~\lor~               &p_1 = 0 \land p'_1 = 0 \land p''_1 \neq 0 \land p_2 < 0 \land q = 0 \land q' < 0 \\
    ~\lor~               &p_1 = 0 \land p'_1 = 0 \land p''_1 \neq 0 \land p_2 < 0 \land q = 0 \land q'=0 \land q'' < 0\,.
\end{align*}

\begin{theorem}
\label{thm:smallchunks}
Let $S$ be a semi-algebraic set encoded either as $\bigwedge_{i=1}^k \bigvee_{j=1}^{m_i} (p_{i j} \bowtie_{i j} 0)$ (DNF) or as $\bigvee_{i=1}^k \bigwedge_{j=1}^{m_i} (p_{i j} \bowtie_{i j} 0)$ (CNF) for some polynomials $p_{i j}$. 
Let $m = \max_i m_i$, $d = \max_{i,j} \deg(p_{i j})$, and $\rho = \max_{i,j} \rk(p_{i j})$.  
Then $\exit_f(S) \lor \exit_{-f}(\lnot S)$ is a union of at most $k m^k \rho (\rho+1)^{k-1}$ basic semi-algebraic sets  
\[
q_1 \bowtie_1 0 \land \dotsc \land q_s \bowtie_s 0\,, 
\]
where $s \leq m - 1 + k(\rho+1)$ and $\deg(q_j) \leq d + \rho(\deg(f)-1)$. 
\end{theorem}
\begin{proof}
Suppose $S \equiv \bigvee_{i=1}^k \bigwedge_{j=1}^{m_i} (p_{i j} \bowtie_{i j} 0)$ (the same reasoning applies when $S$ is in CNF). 
Thus $\lnot S \equiv \bigwedge_{i=1}^k \bigvee_{j=1}^{m_i} \lnot (p_{i j} \bowtie_{i j} 0)$. 
According to Propositions~\ref{prop:DNFsizes} and~\ref{prop:CNFsizes}, $\exit_f(S)$ is a union of at most $k m^k$ basic semi-algebraic sets, each involving $\exit_f(p_{i j} \bowtie_{i j} 0)$, whereas $\exit_{-f}(\lnot S)$ is the union of at most $k m^k$ basic semi-algebraic sets, each involving $\exit_{-f} \lnot(p_{i j} \bowtie_{i j} 0)$. 
If $p_{i j} \bowtie_{i j} 0$ encodes a closed set, then its negation encodes an open set (and vice versa). Thus at least one of the expressions
\[
\exit_f(p_{r s} \bowtie_{r s} 0) \land \bigwedge_{j=1,j\neq s}^{m_r} (p_{r j} \bowtie_{r j} 0) \land \bigwedge_{i=1,i \neq r}^{k} \lnot \INf(p_{i \ell_i} \bowtie_{i \ell_i} 0)
\]
or 
\[
\exit_{-f} \lnot(p_{r s} \bowtie_{r s} 0) \land \bigwedge_{j=1,j\neq s}^{m_r} \INnf(p_{r j} \bowtie_{r j} 0) \land \bigwedge_{i=1,i \neq r}^{k} \lnot(p_{i \ell_i} \bowtie_{i \ell_i} 0)
\]
reduces to \textsf{False} syntactically, and the total number of basic semi-algebraic sets is therefore $k m^k$. Now, according to Lemma~\ref{lem:disjunctDNF}, each of the above expressions is the union of at most $\rho(\rho+1)^{k-1}$ basic semi-algebraic sets (after evaluating $\INf$ and $\exit_f$ for atomic formulas). Thus $\exit_f(S) \lor \exit_{-f}(\lnot S)$ is the union of at most $k m^k \rho(\rho+1)^{k-1}$ basic semi-algebraic sets, as stated. 
The bounds on the total number of the involved polynomials as well as their degrees are direct consequences of Lemma~\ref{lem:disjunctDNF} and of the bound on the total degree of high-order Lie derivatives, namely $d + \rho(\deg(f)-1)$. 
\end{proof}

The main conclusion from the analysis performed in this section is the following: instead of solving one large real quantifier elimination problem which results from a na\"ive application of Theorems~\ref{thm:charbyinduction} and~\ref{thm:char} (coarse granularity), it is instead possible to solve exponentially many (precisely $k m^k \rho(\rho+1)^{k-1}$) smaller real quantifier elimination problems as in Theorem~\ref{thm:smallchunks}; these smaller problems furthermore only involve basic semi-algebraic sets (fine granularity). 

In theory, there exist decision procedures for deciding universally (or existentially) quantified sentences of real arithmetic that have singly exponential worst case complexity $(s d)^{O(n)}$, where $s$ is the number of polynomials, $d$ their maximum degree and $n$ the number of variables~\citep{GRIGOREV198865}. %
Each of the smaller QE problems features fewer polynomials with a lower maximum degree than the original QE problem. 
The potential gain in complexity is however mitigated by the number of these small problems, which is exponential as stated in Theorem~\ref{thm:smallchunks}. 

The procedure $\ExitSet{}$, as defined in Section~\ref{sec:exitsetalg}, seeks a trade-off between the fine and coarse granularities which translates into a trade-off between the computational cost of QE problems versus the number of QE problems to solve. 
Combined with the syntactic reductions to \textsf{False} of $\exit_f(A)$ whenever $A$ is an atomic formula encoding an open set, the concept of exit sets provides a powerful tool from a computational standpoint -- in addition to its ability to  characterize positively invariant sets in full generality as stated in Theorem~\ref{thm:char}. 

The next section provides some examples that are out of reach for \LZZ{} (see Section~\ref{sec:lzzalg}) and where \ExitSet{} (see Section~\ref{sec:exitsetalg}) succeeds in deciding set positive invariance. 
Notice that, although one can divide the QE problem in \LZZ{} into basic semi-algebraic sets, such an approach will not benefit from the syntactic reductions to \textsf{False} offered by $\exit_f$ (without paying an extra computational overhead to detect such cases).  
\section{Experiments}
\label{sec:ex}
For checking positive invariance of sets described by a single atomic formula (e.g. $p<0$), there is no discernible difference in performance between the \LZZ~and \ExitSet~procedures.
However, there is a very palpable difference between the two procedures when checking positive invariance of sets described by more interesting  formulas with non-trivial Boolean structure. The examples below serve to illustrate this difference.
\begin{example}
Consider the non-linear system \mbox{$x'=-x^3$}, \mbox{$~y'=-y^3+x$}. %
To construct a semi-algebraic set with non-trivial Boolean structure, let us consider the sequence of points obtained from a \emph{rational parametrization} of the unit circle $x^2+y^2=1$, e.g. a sequence of points $(x_t,y_t) = (\frac{2 t}{t^2+1}, -\frac{1-t^2}{t^2+1}) \in \mathbb{Q}^2$. From the arithmetic sequence of rational numbers $t_0=-2$, $t_{n+1} = t_n + \frac{1}{8}$ with $t$ in the range $[-2,2]$, we can construct a sequence of half-planes that include the unit disc centred at the origin and are tangent to the unit circle at the points $(x_t,y_t)$. The intersection of these half-planes results in a droplet-like shape shown in Fig.~\ref{fig:ex1} and is characterized by a formula $S$ which is a conjunction of $36$ linear inequalities. 

\begin{figure}[ht]
    \centering
    \begin{subfigure}{.5\textwidth}
    \centering
    \includegraphics[scale=0.2]{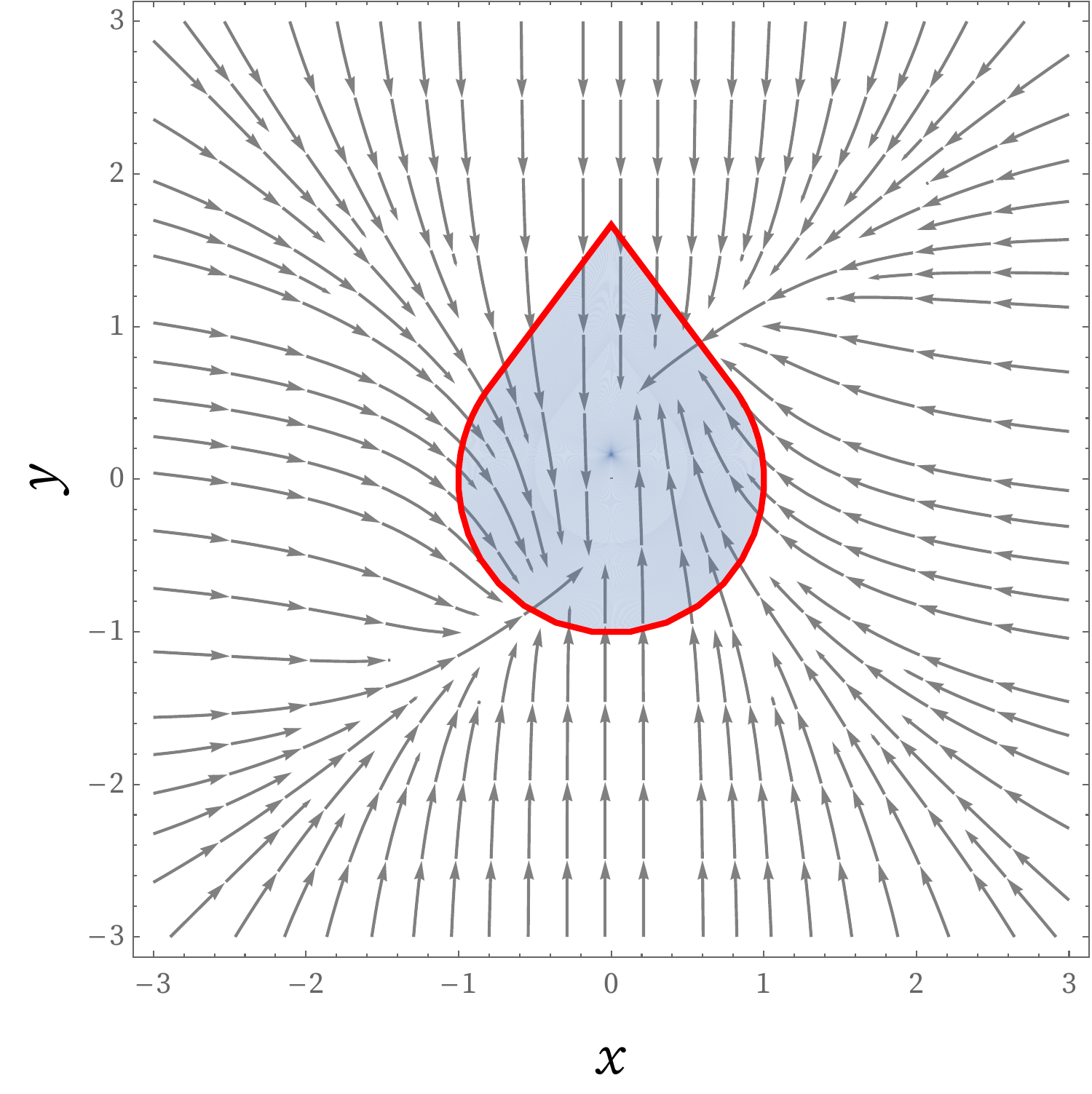}
    \caption{``Droplet'' invariant candidate \label{fig:ex1}}
    \end{subfigure}~\begin{subfigure}{.5\textwidth}
    \centering
    \includegraphics[scale=0.2]{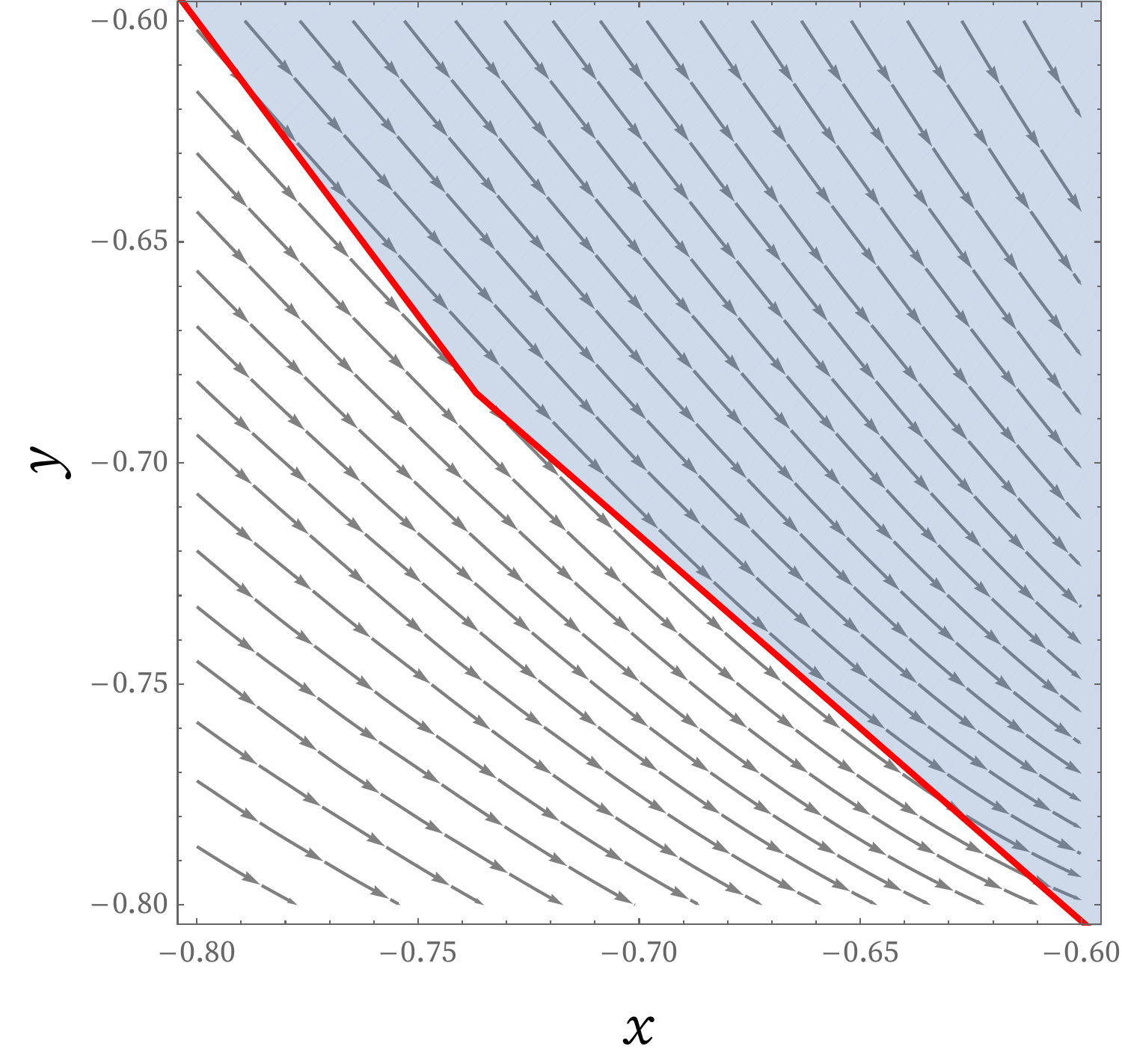}
    \caption{ Flow leaving the droplet \label{fig:cex1} }
    \end{subfigure}
    \caption{Checking positive invariance}
\end{figure}

By inspecting the phase portrait of the system in Fig.~\ref{fig:ex1}, the set defined by this formula appears to be positively invariant, which is something we should be able to check using the procedures described in the previous sections.
Checking positive invariance of $S$ using our implementation of \ExitSet~returns \textsf{False} within 0.3 seconds.\,\footnote{Using Mathematica 12.0, running on a machine with an Intel Core i5-7300U CPU clocked at 2.6GHz with 16GB of RAM.} Indeed, while it is difficult to see from inspecting Fig.~\ref{fig:ex1}, a closer examination (Fig.~\ref{fig:cex1}) reveals that the set characterized by $S$ is not positively invariant because the flow does in fact leave the droplet region.
On the other hand, no answer to this positive invariance question could be obtained using \LZZ~within reasonable time ($>4$ hours). 
\end{example}

\begin{example}
Now let us consider the system $x'=-x^3,~y'=-y^3$ and the set corresponding to the tilted Maltese cross in Figure~\ref{fig:ex2}, which, unlike the previous example, is \emph{not} described by a purely conjunctive formula, but is instead given by a disjunction of $4$ formulas describing the arms of the cross (each arm is described by a formula of the form $p_1\leq 0 \land p_2\leq 0 \land (p_3\leq 0 \lor p_4 \leq 0)$, where each $p_{i=1,2,3,4}$ is linear). For this example, one can verify that the set is indeed a positive invariant using \ExitSet{}, which returns \textsf{True} within 164 seconds. Once more, no answer could be obtained using \LZZ~within reasonable time ($>4$ hours).

\begin{figure}[h!]
    \centering
    \begin{subfigure}{.5\textwidth}
    \centering
    \includegraphics[scale=0.2]{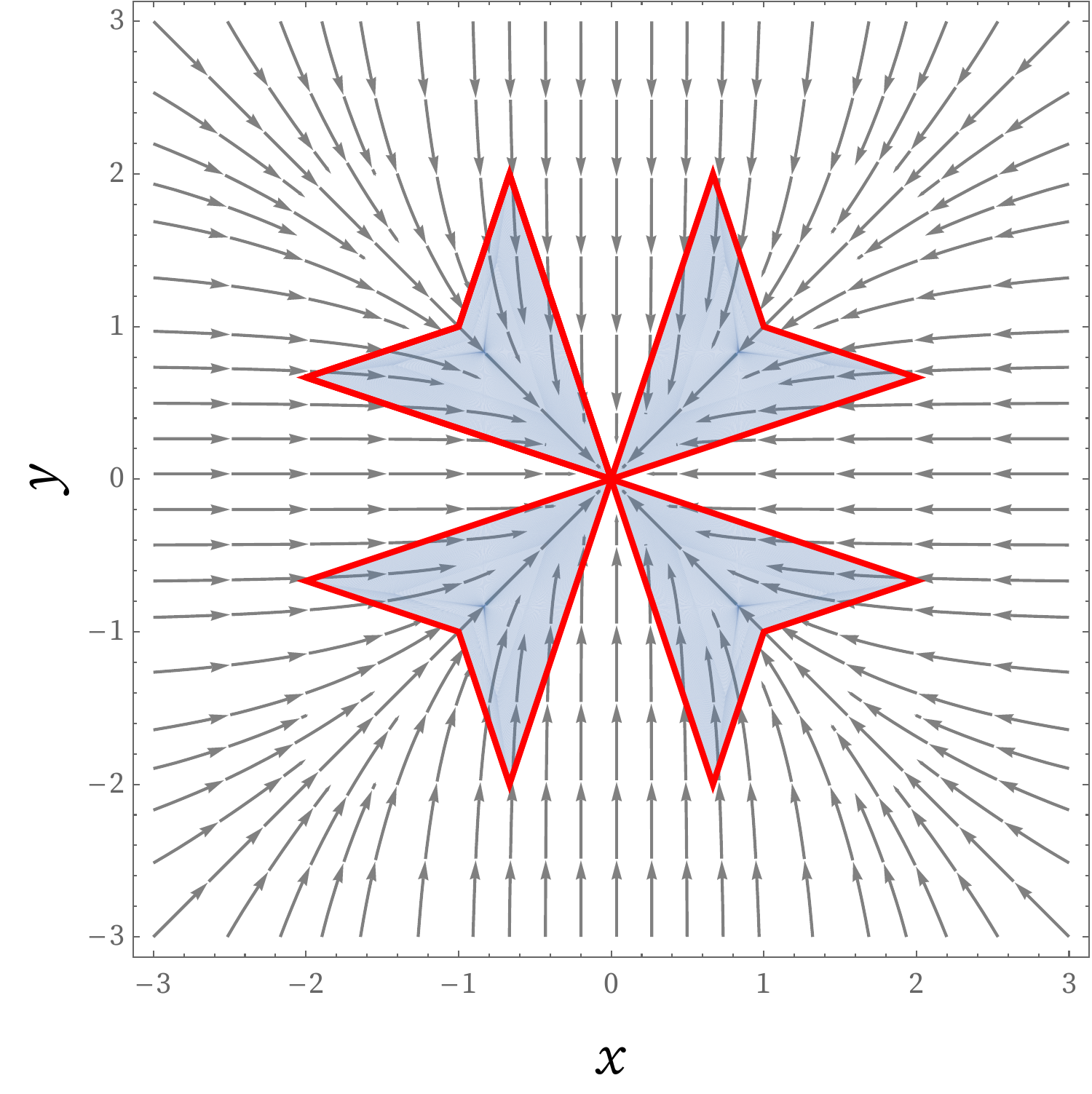}
    \caption{ Semi-linear invariant \label{fig:ex2} }
    \end{subfigure}~\begin{subfigure}{.5\textwidth}
    \centering
    \includegraphics[scale=0.2]{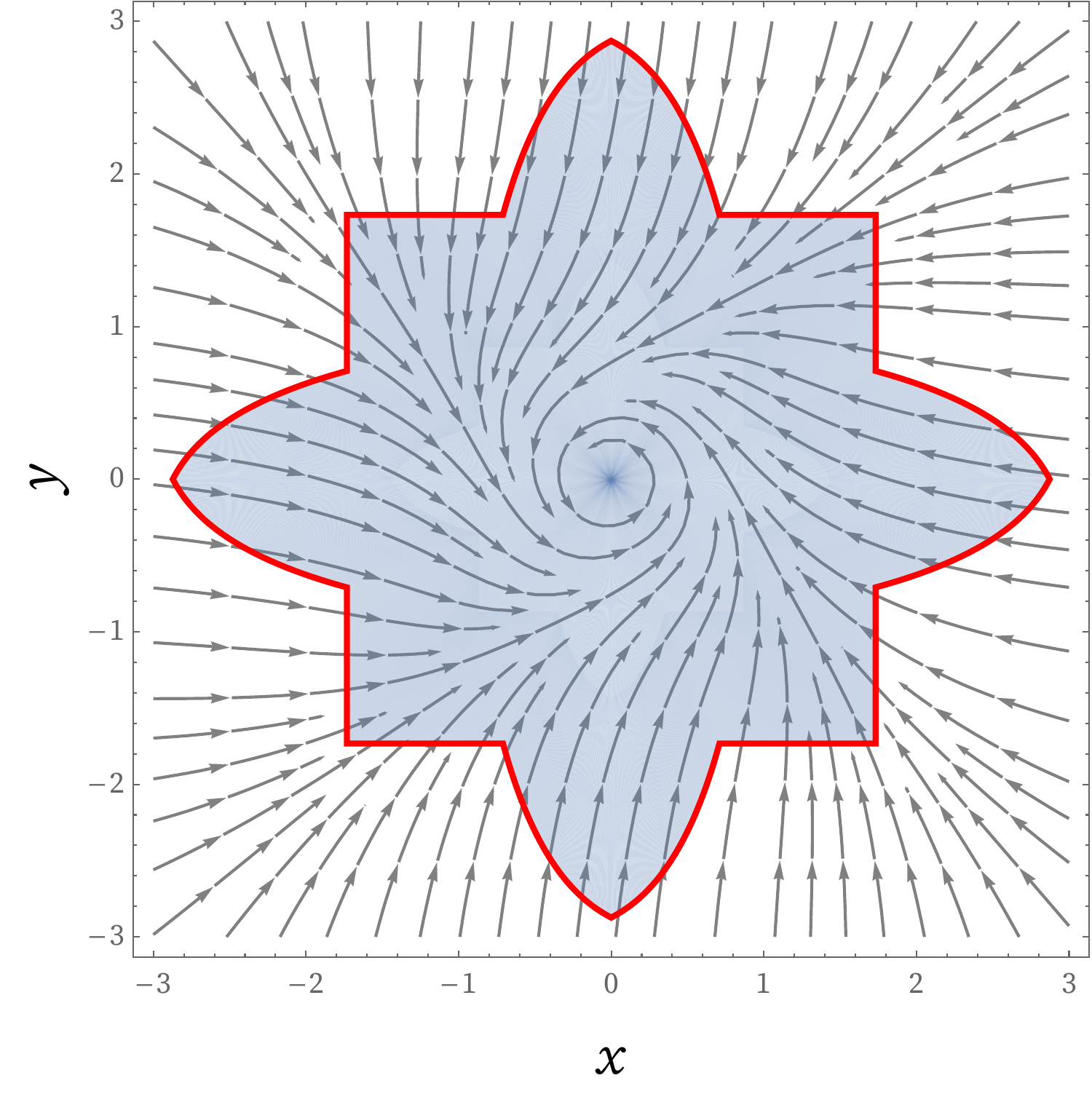}
    \caption{ Semi-algebraic invariant~\label{fig:ex3}}
    \end{subfigure}
    \caption{Positive invariants}
\end{figure}

The set shown in Figure~\ref{fig:ex2} is \emph{semi-linear} because its formal description only features polynomials of maximum degree $1$.
Figure~\ref{fig:ex3} illustrates a semi-algebraic set which is not semi-linear, featuring quadratic polynomials in its formal description; the vector field shown in Figure~\ref{fig:ex3} corresponds to $x'=-x^3 - y,~y'=-y^3 + x$. Using \ExitSet{} we are able to check (within 7 seconds) that the set is indeed positively invariant under the flow of the system, whereas \LZZ~produces the same answer in over 30 minutes.
\end{example}

\section{Positive Invariants Under Constraints}
 In addition to the standard notion of set positive invariance (as given in Definition~\ref{def:posinv}), more general notions have been considered. For example \emph{continuous invariance}, as it is known in the formal verification literature~\citep[see e.g.][]{DBLP:conf/cav/PlatzerC08,DBLP:conf/emsoft/LiuZZ11}, extends positive invariance to accommodate cases in which there is a \emph{constraint} (given by some $Q\subseteq \mathbb{R}^n$) imposed on the evolution of the system. 
\begin{definition}[Continuous invariant]
\label{def:continvar}
A set $S\subseteq \mathbb{R}^n$ is a \emph{continuous invariant} under evolution constraint $Q\subseteq \mathbb{R}^n$ if and only if the following holds:
\[
\forall~{x}\in S.~\forall~t\geq 0.\bigl(\left(\forall~\tau\in [0,t].~{\varphi}(\tau,x)\in Q \right)\to {\varphi}(t,x)\in S \bigr)\,.
\]
\end{definition}
Essentially, in a \emph{continuous invariant} positive invariance is predicated on the constraint $Q$ being maintained.
Thus, positive invariance may be regarded as a special case of \emph{continuous invariance} as defined above, i.e. the special case where the constraint $Q$ is all of $\mathbb{R}^n$. 

\begin{remark}
Readers familiar with temporal logics such as LTL may think of continuous invariance as (very loosely speaking) being in a certain sense analogous to temporal modal operators such as \textsf{Weak Until} $(\mathbf{W})$, i.e. one may think of a continuous invariant described by formula $S$ subject to evolution constraint described by $Q$ as satisfying the temporal logic formula $S~{\mathbf{W}}~\lnot Q$. Of course, the semantics of such a formula needs to be defined over the trajectories of the continuous system rather than discrete traces, e.g. as is done in Signal Temporal Logic~\citep[STL, see][]{DBLP:conf/formats/MalerN04}.
\end{remark}

The work of~\cite{DBLP:conf/emsoft/LiuZZ11} was developed in this slightly more general setting of continuous invariance, rather than positive invariance. A semi-algebraic set $S$ subject to a semi-algebraic evolution constraint $Q$ is a continuous invariant of the system $x'=f(x)$ if and only if~\cite[Thm. 19]{DBLP:conf/emsoft/LiuZZ11}: $S \cap Q \cap \INf(Q)  \subseteq \INf(S)$ and $S^c \cap Q \cap \INnf(Q) \subseteq \INnf(S)^c$.

The \ExitSet~algorithm introduced in this article is likewise easily lifted to check continuous invariance.
\begin{theorem}
A semi-algebraic set $S$ is a continuous invariant for a system of ODEs $x'=f(x)$ subject to a semi-algebraic evolution constraint $Q$ if and only if
\[  \lnot \left(\IsEmpty_f(S,Q \setminus \exit_f(Q))  ~\lor~ \IsEmpty_{-f}(\lnot S, Q \setminus \exit_{-f}(Q)) \right) \enspace .\]
\end{theorem}
The main difference with respect to Theorem~\ref{thm:isemptychar}, is that instead of considering the entire space $\bbr^n$ for both forward and backward flows, we focus on $Q \setminus \exit_{f}(Q))$ (which is equivalent to $Q \cap \INf(Q)$ by Lemma~\ref{lem:link}) for the forward flow and $Q \setminus \exit_{-f}(Q)$ (or equivalently $Q \cap \INnf(Q)$) for the backward flow. These formulations make explicit the fact that the states from which the flow exits $Q$ (formally captured by $\exit_f(Q)$) are not relevant for checking continuous invariance and are thus removed from $Q$. 
Said differently, by construction, the set $\exit_f(Q)$ (resp. $\exit_{-f}(Q))$ is considered a positive (resp. negative) invariant set relative to $Q$.

\subsection{Discrete Abstractions of Continuous Systems}
Problems involving positive invariance checking under evolution constraints (i.e. continuous invariance in the sense of Definition~\ref{def:continvar}) arise frequently in the area of formal verification. Invariants described using formulas with non-trivial Boolean structure are particularly important to verification methods based on \emph{discrete abstractions} of continuous dynamical systems~\citep{DBLP:conf/vmcai/SogokonGJP16}. Briefly, discrete abstraction involves partitioning the state space (e.g. $\mathbb{R}^n$) into disjoint sets that correspond 
to equivalence classes representing states in a discrete transition system. For
example, such a partitioning can be obtained from an \emph{algebraic decomposition}
of $\mathbb{R}^n$ using a finite set of polynomials $\{ p_1,\dots,p_k \}$. Each cell of
this decomposition is described by a conjunction of sign conditions on these polynomials,
e.g. the formula $S \equiv p_1>0 \land p_2=0 \land \dots \land p_k<0$ describes a cell (which is a basic semi-algebraic set corresponding to a single discrete state in the abstraction). Discrete abstractions of continuous systems are obtained by constructing a discrete transition relation between the discrete states. An abstraction is said to be \emph{sound} if the absence of a discrete transition from the state described by $S_i$ to another state described by $S_j$ in the transition relation implies that the continuous system cannot evolve from any state within the set $S_i$ to any state within $S_j$ without leaving the union $S_i \cup S_j$; an abstraction is said to be \emph{exact} if the presence of such a transition implies the existence of a trajectory which starts at a state within $S_i$ and reaches some state in $S_j$ without leaving the union $S_i \cup S_j$ in the process.
In order to construct the transition relation for a sound and exact discrete abstraction one considers the union of neighbouring cells $S_i$ and $S_j$ in the algebraic decomposition and checks whether the set described by $S_j$ is a continuous invariant subject to the constraint $S_i \lor S_j$. There can be
no transition from cell $S_i$ to $S_j$ in the discrete transition relation if 
and only if $S_j$ is continuous invariant under constraint $S_i \lor S_j$ in the sense of Definition~\ref{def:continvar}.
Naturally, the Boolean structure of the formulas involved make the construction of discrete abstractions a potentially fruitful area of application for the \ExitSet~algorithm.

\section{Related Work}\label{sec:relwork}
The method of applying the ascending chain condition to ideals generated by successive Lie derivatives of polynomials in order to prove invariance of algebraic varieties in polynomial vector fields was employed by~\cite{NovikovYakovenko},~\cite{DBLP:conf/tacas/GhorbalP14}, and more recently by~\cite{harms2017polynomial}. \cite{DBLP:conf/emsoft/LiuZZ11} were the first to address positive invariance of semi-algebraic sets using techniques described in Section~\ref{sec:realinduction} of this article. \cite{Dowek2003} investigated the use of real induction to solve kinematic problems involving ODEs.

\cite{DBLP:journals/jacm/PlatzerT20} recently developed a system of formal axioms (one of which formalizes the real induction principle) for reasoning about continuous invariants in differential dynamic logic. This axiomatization is \emph{complete} in the sense that a formal proof of continuous invariance of a \emph{semi-analytic set} represented by a formula $S$ can be derived in differential dynamic logic from the axioms whenever this invariance property holds, and a refutation can be derived whenever it does not.

Among characterizations of positive set invariance in a less general setting than that considered in this article, we note the work of~\cite{CastelanHennet1993}, who reported necessary and sufficient conditions for positive invariance of convex polyhedra in linear vector fields.

\section*{Conclusion}
This article describes two alternative characterizations of positively invariant sets for systems of ODEs with unique solutions. 

The first characterization, along with its associated \LZZ{} decision procedure for checking positive invariance of semi-algebraic sets in poylnomial vector fields, is closely related to the work by~\cite{DBLP:conf/emsoft/LiuZZ11}. 
While the relationship between the work of~\cite{DBLP:conf/emsoft/LiuZZ11} and the principle of \emph{real induction} has been known informally to a number of researchers, this important link has not been adequately elaborated in existing literature. One of our aims in writing this article has been to make this relationship more widely appreciated and also to create an accessible account of the original \LZZ{} decision procedure, along with our own improvements to this method (Section~\ref{subsec:lzzimprovements}) and nuances in its practical implementation informed by our experience (Section~\ref{sec:lzzalg}). 

The second part of the article contributes an alternative characterization of set positive invariance and is based on the notion of exit sets~\citep{conley}. The topological origins of this notion afford certain computational vistas that suggest a very different approach to developing a decision procedure for checking positive invariance than that of~\LZZ{}. The \ExitSet{} procedure developed in Section~\ref{sec:exitsetalg} is, to the authors' knowledge, entirely novel. Its main advantage over \LZZ{} lies in its efficient handling of formulas with non-trivial Boolean structure (a class of problems where the \LZZ{} procedure generally performs poorly). The complexity analysis undertaken in Section~\ref{sec:complexity} sheds some light on the computational advantages of using \ExitSet{}, which is empirically confirmed in a number of examples in Section~\ref{sec:ex}. 

Important topics not touched upon in this article include  \emph{robustness} of positively invariant sets under small perturbations of the system dynamics; indeed, in practical applications, the system of ODEs is often only known approximately and invariants that are not robust are in a certain sense unphysical. In the future we hope to build upon the present work to address these considerations.

\begin{paragraph}{Acknowledgement} 
  The authors are indebted to Dr Paul B. Jackson and Dr Kousha Etessami at the University of Edinburgh for their insights into the real induction principle underpinning the work of Liu, Zhan and Zhao, and extend special thanks to Yong Kiam Tan at Carnegie Mellon University for suggesting improvements to the writing in an early draft of this work (and for bringing the authors' attention to the article~\citep{clark2019} of which they were previously unaware); his formal development of an invariant checking procedure in differential dynamic logic appeared in~\citep{DBLP:conf/lics/PlatzerT18} and~\citep{DBLP:journals/jacm/PlatzerT20}. The authors would also very much like to thank the anonymous reviewers for their careful reading and valuable suggestions for improving the article.
\end{paragraph}

\bibliographystyle{elsarticle-num-names}
\bibliography{root}

\end{document}